\documentclass[11pt]{article}

\usepackage{ifthen}  
\newif\ifExtAbs
 \ExtAbsfalse
\ifExtAbs
\else
\fi

\ifExtAbs
    \usepackage[margin=0.75in]{geometry}
    \usepackage[inline,shortlabels]{enumitem}
\else
    \usepackage[margin=1in]{geometry}
    \usepackage{enumitem}
\fi

\usepackage[utf8]{inputenc}
\usepackage[T1]{fontenc}
\usepackage{microtype}

\usepackage[english]{babel}

\usepackage{charter}
\usepackage{bbm}

\usepackage{tikz}
\usetikzlibrary{arrows.meta,calc,positioning,shadows}
\usepackage{pgfplots}
\pgfplotsset{compat = 1.18}

\usepackage{authblk}  

\usepackage{csquotes,xpatch}

\usepackage[nottoc]{tocbibind}
\usepackage{amsfonts}
\usepackage{amsmath}
\usepackage{amssymb}
\usepackage{amsthm}
\allowdisplaybreaks[2]
\usepackage{mathtools}
\usepackage{thmtools}
\usepackage{booktabs}
\usepackage{macro}

\usepackage{hyperref}
\hypersetup{colorlinks=true,
  breaklinks=true,
  linkcolor=[rgb]{0, 0.2, 0.5},
  citecolor=[rgb]{0, 0.35, 0},
  urlcolor=[rgb]{0, 0.2, 0.5},
pdfpagemode=UseOutlines}
\usepackage[nameinlink]{cleveref}
\crefmultiformat{equation}{eqs.~#2(#1)#3}{ and~#2(#1)#3}{, #2(#1)#3}{ and~#2(#1)#3}
\Crefmultiformat{equation}{Eqs.~#2(#1)#3}{ and~#2(#1)#3}{, #2(#1)#3}{ and~#2(#1)#3}
\crefmultiformat{theorem}{theorems.~#2(#1)#3}{ and~#2(#1)#3}{, #2(#1)#3}{ and~#2(#1)#3}
\Crefmultiformat{theorem}{Theorems.~#2(#1)#3}{ and~#2(#1)#3}{, #2(#1)#3}{ and~#2(#1)#3}
\crefmultiformat{proposition}{propositions.~#2(#1)#3}{ and~#2(#1)#3}{, #2(#1)#3}{ and~#2(#1)#3}
\Crefmultiformat{proposition}{Propositions.~#2(#1)#3}{ and~#2(#1)#3}{, #2(#1)#3}{ and~#2(#1)#3}
\crefmultiformat{corollary}{corollaries.~#2(#1)#3}{ and~#2(#1)#3}{, #2(#1)#3}{ and~#2(#1)#3}
\Crefmultiformat{corollary}{Corollaries.~#2(#1)#3}{ and~#2(#1)#3}{, #2(#1)#3}{ and~#2(#1)#3}
\crefmultiformat{lemma}{lemmas.~#2(#1)#3}{ and~#2(#1)#3}{, #2(#1)#3}{ and~#2(#1)#3}
\Crefmultiformat{lemma}{Lemmas.~#2(#1)#3}{ and~#2(#1)#3}{, #2(#1)#3}{ and~#2(#1)#3}
\crefmultiformat{conjecture}{conjectures.~#2(#1)#3}{ and~#2(#1)#3}{, #2(#1)#3}{ and~#2(#1)#3}
\Crefmultiformat{conjecture}{Conjectures.~#2(#1)#3}{ and~#2(#1)#3}{, #2(#1)#3}{ and~#2(#1)#3}
\crefmultiformat{result}{results.~#2(#1)#3}{ and~#2(#1)#3}{, #2(#1)#3}{ and~#2(#1)#3}
\Crefmultiformat{result}{Results.~#2(#1)#3}{ and~#2(#1)#3}{, #2(#1)#3}{ and~#2(#1)#3}
\crefmultiformat{fact}{facts.~#2(#1)#3}{ and~#2(#1)#3}{, #2(#1)#3}{ and~#2(#1)#3}
\Crefmultiformat{fact}{Facts.~#2(#1)#3}{ and~#2(#1)#3}{, #2(#1)#3}{ and~#2(#1)#3}
\crefmultiformat{question}{questions.~#2(#1)#3}{ and~#2(#1)#3}{, #2(#1)#3}{ and~#2(#1)#3}
\Crefmultiformat{question}{Questions.~#2(#1)#3}{ and~#2(#1)#3}{, #2(#1)#3}{ and~#2(#1)#3}
\crefmultiformat{definition}{definitions.~#2(#1)#3}{ and~#2(#1)#3}{, #2(#1)#3}{ and~#2(#1)#3}
\Crefmultiformat{definition}{Definitions.~#2(#1)#3}{ and~#2(#1)#3}{, #2(#1)#3}{ and~#2(#1)#3}
\crefmultiformat{example}{examples.~#2(#1)#3}{ and~#2(#1)#3}{, #2(#1)#3}{ and~#2(#1)#3}
\Crefmultiformat{example}{Examples.~#2(#1)#3}{ and~#2(#1)#3}{, #2(#1)#3}{ and~#2(#1)#3}
\crefmultiformat{remark}{remarks.~#2(#1)#3}{ and~#2(#1)#3}{, #2(#1)#3}{ and~#2(#1)#3}
\Crefmultiformat{remark}{Remarks.~#2(#1)#3}{ and~#2(#1)#3}{, #2(#1)#3}{ and~#2(#1)#3}

\theoremstyle{plain}
\newtheorem{theorem}{Theorem}
\newtheorem*{theorem*}{Theorem}
\newtheorem{proposition}[theorem]{Proposition}
\newtheorem*{proposition*}{Proposition}

\newtheorem*{corollary*}{Corollary}
\newtheorem{lemma}[theorem]{Lemma}
\newtheorem*{lemma*}{Lemma}
\newtheorem{conjecture}{Conjecture}
\newtheorem*{conjecture*}{Conjecture}

\newtheorem*{result*}{Result}
\newtheorem{fact}{Fact}
\newtheorem*{fact*}{Fact}

\newtheorem*{question*}{Question}
\theoremstyle{definition}
\newtheorem{definition}{Definition}
\newtheorem*{definition*}{Definition}
\newtheorem{example}{Example}
\newtheorem*{example*}{Example}
\theoremstyle{remark}
\newtheorem{remark}{Remark}
\newtheorem*{remark*}{Remark}

\newcommand{\val}{\mathrm{value}}
\allowdisplaybreaks

\crefname{namedthm}{Theorem}{Theorems}
\Crefname{namedthm}{Theorem}{Theorems}

\title{Keyless secrecy against bounded adversaries}

\ifExtAbs
    \author{}
\else
    \author{Anne Broadbent\thanks{anne.broadbent@uottawa.ca}~}
    \author{Upendra Kapshikar\thanks{ukapshik@uottawa.ca}}
    \affil{Department of Mathematics and Statistics, University of Ottawa, Canada}
    \author{Denis Rochette\thanks{denis.rochette@inria.fr}}
    \affil{CPHT, LIX, CNRS, Inria, École polytechnique, Institut Polytechnique de Paris, Palaiseau, France}
\fi

\date{}

\begin{document}

\maketitle

\ifExtAbs
\else
    \begin{abstract}
        We introduce a keyless coding/cryptographic primitive that asks for two guarantees at once: the receiver is never fooled into accepting a message other than the one sent, and the adversary learns nothing about the message unless the receiver aborts.
        Neither the sender nor the receiver holds a secret key, and no computational hardness is assumed.
        The only restriction is on the adversary's online computation: the receiver aborts if nothing arrives by a fixed deadline, so the adversary must forward something before it, and her map in that window is computed by a circuit of size (or depth) at most $p$; before and after, she is unbounded.
        For every polynomial~$p$ we construct such an efficient quantum scheme, encoding $k$-bit messages into $n = O(k)$ qubits with circuits of size $\poly(n,p)$.
        No classical scheme achieves this type of \emph{everlasting security}, at any choice of parameters.
        
        Our techniques also resolve open questions about universal tamper detection against families restricted in cardinality rather than in circuit size.
        We settle a question of Broadbent, Kapshikar and Rochette on relaxed tamper detection.
        As a corollary, we obtain the first efficient non-malleable code secure against global quantum tampering, with no split-state restriction, whereas the current constructions in the literature require the codeword to be split into non-communicating shares.
    \end{abstract}
    \newpage 
    \tableofcontents
    \newpage 
\fi

\ifExtAbs
    \ifExtAbs
    \else
    \section{Introduction}
\fi

Consider the following scenario: Alice wants to send a message to Bob. 
They have access to a channel. 
The channel is public and insecure. 
An adversary, Eve, can monitor what is being sent through the channel.
She can read it and even tamper with the messages sent by Alice.
Alice and Bob have access to randomness, but the randomness is not private. 
Eve too can read this randomness. 
A typical experiment in this model can be set up as follows:
Fix some message set $\M$.
\begin{enumerate}
\setcounter{enumi}{-1}
    \item Alice and Bob get their public random string $\pp$.
    \item \emph{Encoding:} Alice chooses $m \in \M$ to be sent to Bob. 
    She encodes it using an encoder:
    $c = \enc(\pp,m)$ and sends it to Bob.
    \item \emph{Attack Phase 1:} Eve intercepts $c$ and creates a two-part output $(\widehat{c},e) = f_1(\pp, c)$. 
    Eve keeps $e$ for herself and sends $\widehat{c}$ to Bob.
    \item \emph{Decoding:} Bob decodes $\widehat{c}$ to get $\widehat{m} = \dec(\pp,\widehat{c})$ where $\widehat{m} \in \M \bigcup \lbrace \perp \rbrace$, with $\perp$ indicating the event that Bob aborts, detecting Eve's presence.
    \item \emph{Attack Phase 2:} Eve post-processes $e$ to get her candidate guess for $m$, given by $m_e= f_2(\pp,e)$.
\end{enumerate}

We will allow the codewords, the modified codewords,  and Eve's retained information to be quantum, at which point they will be assigned their own registers ($C, \widehat{C}$, and $E$, respectively). {The public random string $\pp$, and message $m$, and the decoder output $\widehat{m}$ remain classical.}

\ifExtAbs
    For a parameter $p$, let $\mathcal{F}_1(p)$ denote the class of quantum circuits of size at most $p$ over a fixed universal gate set. We restrict Eve's Phase~1 map $f_1$ to $\mathcal{F}_1(p)$, while her Phase~2 map $f_2$ is unrestricted. We require the scheme to satisfy the following three properties, where $\varepsilon$ denotes a negligible error bound in the security guarantees:
    \begin{description}
        \item[\normalfont(C) Completeness.]
        In the absence of tampering, Bob always recovers the message sent by Alice.
    
        \item[\normalfont(D) Tamper detection.]
        For every Phase~1 attack $f_1\in\mathcal{F}_1(p)$, Bob either recovers the message sent by Alice or aborts, except with probability at most $\varepsilon$.
    
        \item[\normalfont(ES) Everlasting secrecy.]
        For every Phase~1 attack $f_1\in\mathcal{F}_1(p)$, every unrestricted Phase~2 attack $f_2$, and every pair of messages $m_0,m_1\in\M$, Eve's views corresponding to Bob not aborting are $\varepsilon$-indistinguishable.
    \end{description}
\else
    From such a scheme, Alice and Bob typically want two guarantees:
    \begin{enumerate}
        \item Detection (D): Bob is not fooled into accepting an incorrect message. 
        \[ \Pr \left( \widehat{m} \neq m \text{ and }\widehat{m} \neq \perp \right) \leq \varepsilon, \qquad \text{ for small small enough parameter $\varepsilon$}. \]
    \item Everlasting Secrecy (ES): the probability that Eve learns $m$ and Bob does not abort is small\footnote{In this work, we use a distinguisher-based security notion, which is stronger than this. Informally, distinguisher-based security shows that for any pair of messages $(m_1,m_2)$, the adversary's view throughout the experiment is indistinguishable.}.
    \end{enumerate}
\fi

Detection and secrecy, in one form or another, are what most of cryptography is built to provide, and what separates one primitive from another is the strength of the guarantees it provides and the resources it assumes.
Symmetric-key encryption provides unconditional secrecy, and a message authentication code provides detection, but both rely on a secret key that Alice and Bob already share, which Eve does not.
Public-key encryption removes the necessity of pre-shared private keys at the cost of some computational assumption, typically the hardness of a certain computational problem, such as factoring, discrete logarithm or learning with errors.
The security of such a scheme lasts only as long as the assumption does. 
Moreover, a ciphertext recorded today may be opened by a better algorithm discovered in the future. 
Quantum key distribution~\cite{BB84} comes closest to assuming nothing: it delivers unconditional guarantees without a long shared key, but it requires an authenticated classical channel and several rounds of interaction.
On the other end of the picture, the bounded-storage model~\cite{Mau92,CM97} buys everlasting secrecy from a bound on Eve's memory, while tamper-detection and non-malleable codes~\cite{DPW18,JW15} give keyless detection guarantees while restricting Eve to a fixed class of tampering families.

Each of these primitives fixes a resource: a shared secret, a hardness assumption, an authenticated channel, a memory bound, a restricted tampering class, and provides detection, secrecy, or both, in exchange.
Consider the keyless scenario we described above. 
The scenario above fixes none of them.
In this work, we consider what can be achieved there.
Concretely,
\begin{center}
\emph{In the scenario above, where Alice and Bob share no secret and make no computational hardness \\ assumptions, what is the weakest restriction on Eve under which detection and\\ secrecy can be achieved at once?}
\end{center}

\ifExtAbs
\else
    Some restrictions on Eve are necessary, as we now explain.
    Suppose $f_1$ is unrestricted.
    The scheme is public and so is $\pp$.
    So Eve can run the decoder herself.
    She computes $m = \dec(\pp,c)$, encodes it afresh as $\widehat{c} = \enc(\pp,m)$, and forwards
    $\widehat{c}$ to Bob.
    She keeps $e = m$.
    Bob receives a perfectly legitimate encoding of the message Alice sent, so he accepts and outputs
    $\widehat{m} = m$.
    Bob has no indication that anything occurred, and secrecy fails with probability $1$. 
    
    This obstruction is not special to this attack or to the scheme.
    Bob and Eve stand in the same relation to the codeword.
    Both hold a codeword, both know $\pp$, and both know the scheme.
    Anything Bob can do to recover $m$, Eve can do as well.
    A scheme can therefore only work if something separates them.
    Symmetric-key encryption separates them with a secret, and public-key encryption with a trapdoor.
    In the keyless setting, neither is available, so the separation must come from elsewhere.
    We separate them by computational power.
    We assume that Eve's phase-1 map~$f_1$ is computed by a circuit of size at most~$p$.
    This is the only assumption we make.
    Note that it is not a hardness assumption: nothing is conjectured to be difficult, and no problem is assumed to be intractable.
    Neither is this an assumption of any cryptographic structure, such as the existence of one-way functions or pseudo-random unitaries.
    It is a bound on a resource rather than a hardness assumption, and in that respect its closest relative is the bounded-storage model~\cite{Mau92,CM97} and its quantum counterpart~\cite{DFSS05}, where the constrained resource is memory (or space).
    A storage bound is a claim about Eve's apparatus, and Alice and Bob can neither verify nor necessarily enforce it.
    It also weakens on its own: as we develop better quantum memory, the bounded-storage assumption may no longer hold.
    A timing bound is enforced by the protocol itself.
    Bob aborts if nothing has arrived by time~$t$, so Eve has no option but to act inside that window, which is phase 1 of our model.
    Such a \emph{timed deadline} motivates a bound on the depth or size of the circuit: in time $t$, Eve can apply only so many layers of gates. Since a depth-$p$ circuit on $n$ qubits has size $O(pn)$, and our results hold for every polynomial bound, stating them in terms of size covers the depth-bounded adversary as well.

    As mentioned before, we leave phase 2 unrestricted, as the timing motivation suggests. The clock runs only while Eve holds the codeword; once she has forwarded $\widehat{c}$, she may keep $e$ for as long as she wishes and apply an arbitrary $f_2$ to it, with no bound on size, depth, or memory.
    The assumption binds her only during phase 1, while the conclusion that $e$ reveals nothing about $m$ holds forever after.
    Such secrecy is often referred to as the \emph{everlasting secrecy}.
    
    With the model in place, we can now state the question concretely:
    
    \noindent 
    Fix a universal gate set, and for a parameter $p$ let $\mathcal{F}_1(p)$ be the class of circuits\footnote{We impose no uniformity: Eve's precomputation is free, so she may choose her circuit as an arbitrary function of $\pp$.} over that gate set with size at most $p$.\footnote{Since a depth-$p$ circuit on $n$ wires has $O(pn)$ gates, bounds stated for size cover depth as well.}

    \begin{center}
    \emph{For every polynomial $p$, does there exist a scheme of size $q=\poly(p)$ achieving detection and secrecy against all of $\mathcal{F}_1(p)$?}
    \end{center}
    
    Note that the order of quantifiers is forced.
    No single scheme can be secure against $\mathcal{F}_1(p)$ for every $p$, since the re-encoding attack costs one evaluation of $\dec$ followed by one of $\enc$, and a fixed scheme fixes that cost.
    Thus, necessarily $q > p$: the honest parties must out-compute the adversary, which is what the time deadline of phase 1 supplies.
    
    One might expect the gap to be enough on its own, since Eve cannot afford to run $\dec$ at all. 
    It is worth noting that out-computing Eve (that is assuming $q>p$) in Phase 1 is necessary, but it is not sufficient. We support this with two examples.
    
    Suppose the scheme was classical, that is, the encoder and decoder are classical randomized maps.
    No classical scheme can achieve what we need for reasons similar to the impossibility of unclonable encryption with a classical scheme~\cite{BL20}.
    Eve copies $c$ into $e$ and forwards the original untouched.
    Bob receives exactly what Alice sent, so he accepts and outputs $\widehat{m} = m$; detection is not violated.
    After the deadline, Eve computes $\dec(\pp,e) = m$ at her leisure, since phase 2 is unbounded.
    Copying $n$ bits costs $n$ gates, so this attack lies in $\mathcal{F}_1(p)$ for every $p \geq n$.
    Thus, no classical scheme achieves secrecy, at any $p$ and any $q$, however large the gap between them.
    
    One might wonder whether copying is the entire obstruction, and whether making the codeword quantum, on top of $q > p$, would be enough.
    Hence, in the second example we argue that it is not.
    Consider conjugate coding, the encoding behind BB84 and Wiesner's quantum money~\cite{Wie83}, which sends $\bigotimes_i H^{\theta_i}\ket{x_i}$ for a classical string $x$ and a basis vector $\theta$.
    Its codewords are unclonable, but unclonability is a property of states whose preparation basis is unknown to the adversary, and here $\theta$ is part of $\pp$.
    Eve applies a CNOT only for qubits with $\theta_i=0$ into a fresh ancilla.
    With very high probability over $\pp$, there are at least $\sqrt{n}$ qubits encoded in computational basis, that is, $\theta_i=0$. 
    On those wires $\ket{x_i}$ is a basis state, so $\ket{x_i}\ket{0} \mapsto \ket{x_i}\ket{x_i}$: the forwarded qubit is untouched, Bob accepts, and Eve keeps a copy of $x_i$ for Phase 2.
    Nothing was cloned outside a known basis and nothing was measured.
    The attack costs one gate per wire against a decoder of $q=\Theta(n)$ gates, so it sits well inside $q > p$; restricted to $\sqrt{n}$ wires it costs $p=\sqrt{n}$ gates and yields $\sqrt{n}$ bits, hence the scheme is not indistinguishable secure.
\fi

\ifExtAbs
\else
    We now describe the primitive we want to construct. (see \Cref{sec:our_primitive} for formal definitions).
    We say that $(\enc,\dec)$ is a \emph{$(p,q)$-tamper detection scheme} if $\enc$ and $\dec$ have size $q$ as quantum circuits and the following three clauses hold for a negligible $\varepsilon$.
    Let $\pp$ be the public random string.
    \begin{description}
    \item[\normalfont(C) Completeness.] For every $\pp$ and every $m \in \M$, an untampered codeword decodes correctly: $\dec(\pp,\enc(\pp,m)) = m$ with probability $1$.
    \item[\normalfont(D) Tamper detection.] For every $f_1 \in \mathcal{F}_1(p)$ and every $m$, $\Pr[\widehat{m} \notin \{m,\perp\}] \leq \varepsilon$.
    \item[\normalfont(ES) Everlasting secrecy.] For every $f_1 \in \mathcal{F}_1(p)$, every (unbounded) $f_2$, and every pair of messages $m_0,m_1 \in \M$, weighted on Bob not aborting, Eve's views in the experiments with messages $m_0$ and $m_1$ are $\varepsilon$-indistinguishable.
    \end{description}
\fi

\subsection{Our results}\label{sec:results}

Throughout, $k$ denotes the message length, $n = \gamma k$ the number of physical qubits, $d = 2^n$, and $\gamma$ is called the expansion factor of the scheme.
Let $\M = \{0,1\}^k$ and split the $n$ qubits into a bipartition as $A \bigsqcup B$ with $\dim A = 2^k$ and $\dim B = 2^{n-k}$. 

\ifExtAbs
    The encoding is the one introduced by~\cite{BKR26}. For a unitary $U \in \mathcal{U}(d)$, described by the public parameter $\pp$, define
    \[
      \Enc_U(m)
      = U\left(\ket{m}_A \otimes \ket{0^{n-k}}_B\right),
      \qquad m\in\M.
    \]
    To decode, $\Dec_U$ applies $U^{*}$, measures $B$ in the computational basis and aborts unless the outcome is $0^{n-k}$; otherwise, it measures $A$ and returns the result. Thus, $\Dec_U(\Enc_U(m))=m$ with probability~$1$ for every $U$ and every $m$. Our results differ in the distribution from which $U$ is drawn.
\else
    The encoding is the one introduced by Broadbent, Kapshikar and Rochette~\cite{BKR26}. 
    For a unitary $U \in \mathcal{U}(d)$, described via the public parameter $\pp$, set
    \[
      \Enc_U(m) = U\left(\ket{m}_A \otimes \ket{0^{n-k}}_B\right), \qquad \text{for } m \in \M,
    \]
    and let $\Dec_U$ be the measurement $\{\Pi_t\}_{t \in \M} \cup \{\Pi_\perp\}$ given by
    \[
      \Pi_t = U\big(\proj{t}_A \otimes \proj{0^{n-k}}_B\big)U^{*}, \qquad
      \Pi_\perp = \id - \sum_{t \in \M} \Pi_t .
    \]
    Equivalently, $\Dec_U$ applies $U^{*}$, measures $B$ in the computational basis, aborts if the outcome is nonzero, and otherwise measures $A$ and returns the result.
    Completeness of the scheme is immediate: $\Dec_U(\Enc_U(m)) = m$ with probability $1$ for every $U$ and every $m$.
    What differs across our results is the distribution from which $U$ is drawn.
\fi
For \Cref{thm:B} and \Cref{thm:C}, $U$ is Haar random, and the statements are existential: they assert that all but a negligible fraction of unitaries work.
For \Cref{thm:A}, $U$ is drawn from an $\varepsilon$-approximate unitary $t$-design for $t = \poly(n,p)$, which is samplable by a circuit of size $\poly(n,p)$~\cite{MPSY24} and is what makes $\Setup$, $\Enc$ and $\Dec$ efficient.
Completeness holds for every unitary in the support, so no clause is lost in the substitution.

\paragraph{Main result.}
Our central theorem efficiently constructs a scheme against every polynomial-sized adversary.
Here $q = q_{\enc} + q_{\dec}$ denotes the combined size of the encoder and decoder.

\ifExtAbs
    \begin{namedthm}[$(p,q)$-tamper detection]\label{thm:A}
\else
    \begin{namedthm}[$(p,q)$-tamper detection; \Cref{thm:pq-main}]\label{thm:A}
\fi
Fix a constant $\gamma > 2$.
For every polynomial $p$ there is a polynomial $q = \poly(n,p)$ such that the encoding above, with $U$ drawn from an $\eps$-approximate $2\ell$-design of order
\[
  \ell = \big\lceil k + \log_2|\mathcal{F}_1(p)| + n \big\rceil,
  \qquad \log_2(1/\eps) = \poly(n,p),
\]
is a $(p,q)$-tamper detection scheme.
\end{namedthm}

\ifExtAbs
    To our knowledge, this is the first keyless coding scheme whose guarantee combines detection with secrecy against a tampering adversary.
\else
    To our knowledge, this is the first keyless coding scheme whose guarantee combines detection with secrecy against a tampering adversary.
    As mentioned before, keyed constructions such as authenticated encryption achieve both, but from a shared secret, and by the copying attack above, no classical scheme achieves both at any parameter setting.
    It is complete, both \textup{(D)} and \textup{(ES)} hold with error $2^{-\Omega(n)}$, and $\Setup$, $\enc$, $\dec$ are implementable by circuits of size $q$, necessarily with $q > p$.
\fi

\paragraph{Consequences for cardinality-bounded families.}
\Cref{thm:B,thm:C} concern a different object, so we set out the definitions before stating them.
Here the adversary is a single channel $\Phi$ drawn from a family $\Fadv$ fixed in advance, acting on the whole codeword.
Unlike $\mathcal{F}_1(p)$, the family is not constrained computationally: it may contain channels of arbitrary circuit size, and is restricted only in cardinality.
Nor does the adversary retain a register, so there is no phase 2 and no secrecy requirement; what is asked is only that Bob not be deceived.

\ifExtAbs
    \begin{definition}[tamper detection \cite{BK23, BB23, BKR26}]
        Let $\Fadv$ be a family of channels on $\CC^d$.
        A code is $\eps$-\emph{tamper detecting} against $\Fadv$ if, for every message $m$ and every nontrivial $\Phi\in\Fadv$, the decoder aborts with probability at least $1-\eps$ after $\Phi$ acts on the codeword.
        It is $\eps$-\emph{relaxed tamper detecting} if, except with probability $\eps$, the decoder either outputs the original message $m$ or aborts.
        When the code depends on a random unitary $U$, we say that the property holds \emph{with overwhelming probability} if it holds for a $1-\negl(k)$ fraction of $U$.
    \end{definition}
\else
    \begin{definition}[tamper detection \cite{BK23, BB23, BKR26}]\label{def:td-family}
    Let $\Fadv$ be a family of channels on $\CC^{d}$ and let $\eps > 0$.
    A code $(\enc,\dec)$ is $\eps$-\emph{tamper detecting} against $\Fadv$ if for every $m \in \M$ and every $\Phi \in \Fadv$ acting nontrivially on the code space,
    \[
      \Pr\big[\dec(\Phi(\enc(m))) = \perp\big] \ \ge\ 1 - \eps,
    \]
    and $\eps$-\emph{relaxed tamper detecting} against $\Fadv$ if
    \[
      \Pr\big[\dec(\Phi(\enc(m))) \in \{m,\perp\}\big] \ \ge\ 1 - \eps .
    \]
    When the code depends on a unitary $U$, we say the property holds \emph{with overwhelming probability} if it holds for a $1 - \negl(k)$ fraction of $U$.
    \end{definition}
\fi

\ifExtAbs
    The relaxed notion allows the decoder to return either the original message or $\perp$, and is the natural variant for which a cardinality bound alone can suffice. Two parameters appearing in earlier work are the \emph{Kraus rank} of the tampering channel and its \emph{entanglement fidelity} $\Fe(\Phi)$, which measures closeness to the identity. Some separation from the identity is necessary, since the identity channel cannot be detected, whereas a Kraus-rank bound is an additional restriction that we remove.

    In our previous work~\cite{BKR26}, we studied the same encoding against families $\Fadv$ bounded in cardinality rather than circuit size. Our main result there assumed both $|\Fadv|\le 2^{d^\alpha}$ and a Kraus-rank bound. We conjectured that cardinality alone suffices for relaxed detection throughout $\alpha<1$, and asked whether the rank assumption could be removed for full detection. The present work resolves both questions.
\else
    The relaxed notion permits the decoder to return the correct message even when the codeword was tampered with, and asks only that it never return a different one.
    For classical tamper detection code, a constant tampering function sends every codeword to a fixed string, which decodes to some fixed message, so relaxed detection fails outright for all but one message.
    Quantumly it is the natural weakening, and it is the one for which the cardinality bound alone turns out to suffice.
    
    Two further quantities appear in the hypotheses of earlier works.
    The \emph{Kraus rank} of $\Phi$ is the least number of operators in a Kraus representation $\Phi(\rho) = \sum_i K_i \rho K_i^{*}$, and measures how much ancilla the channel uses.
    The \emph{entanglement fidelity} $\Fe(\Phi) \in [0,1]$ measures how close $\Phi$ is to the identity, with $\Fe(\Phi) = 1$ exactly when $\Phi = \id$.
    Some separation from the identity is unavoidable in any detection statement, since a channel that does nothing cannot be detected; a bound on Kraus rank, by contrast, is a restriction on the adversary that we remove in this work.

    Broadbent, Kapshikar and Rochette~\cite{BKR26} studied the same encoding against a differently restricted adversary: a single channel drawn from a family $\Fadv$ that is bounded in cardinality rather than in circuit size.
    Most of our analysis transfers to their setting, because it bounds the behaviour of an arbitrary channel on the encoding and never uses the fact that Eve is a small circuit.
    Their main theorem constructs tamper-detection codes under two hypotheses: a cardinality bound $|\Fadv| \le 2^{d^{\alpha}}$, and a Kraus-rank bound $\rank(\Phi) \le d^{1-\delta}$, coupled by $2\alpha < \delta$.
    They conjecture that the cardinality bound alone suffices for \emph{relaxed} detection, in which the decoder is required only never to return a message other than the one sent, throughout the range $\alpha < 1$; and they ask whether the rank hypothesis can be removed from the full statement.
    Our methods settle both questions.
\fi

\ifExtAbs
    \begin{namedthm}[universal relaxed tamper detection]\label{thm:B}
\else
    \begin{namedthm}[universal relaxed tamper detection; \Cref{thm:utd}]\label{thm:B}
\fi
Fix $\alpha < 1$ and $\gamma > \tfrac{2}{1-\alpha}$, and let $\Fadv$ be any family of channels on $\CC^{d}$ with $|\Fadv| \le 2^{d^{\alpha}}$.
For all sufficiently large $n$, a Haar-random $U$ yields a $2^{-k}$-relaxed-tamper-detecting code against $\Fadv$ with probability $1 - \negl(k)$.
\end{namedthm}

This proves the conjecture of~\cite{BKR26} in its full range.
It was previously known only for structured families, such as replacement and measure-and-prepare channels, where the specific form of the channel can be exploited directly.

\ifExtAbs
    \begin{namedthm}[rank-free tamper detection]\label{thm:C}
\else
    \begin{namedthm}[rank-free tamper detection; \Cref{thm:fullTDfinal}]\label{thm:C}
\fi
Fix $\alpha < \tfrac12$ and $\delta > 0$, let $|\Fadv| \le 2^{d^{\alpha}}$ with $\Fe(\Phi) \le d^{-\delta}$ for every $\Phi \in \Fadv$, and let $\gamma > \max\{\tfrac{2}{1-\alpha}, \tfrac1\delta, \tfrac{1}{1-2\alpha}\}$.
For all sufficiently large $n$, a Haar-random $U$ yields a $\Theta(2^{-k})$-tamper-detecting code against $\Fadv$ with probability $1 - \negl(k)$.
\end{namedthm}

Against the main theorem of~\cite{BKR26} the change is that the Kraus-rank hypothesis is deleted, and with it the coupling $2\alpha < \delta$ between the cardinality and rank parameters.
The entanglement-fidelity hypothesis remains, and unlike the rank hypothesis it is not a bottleneck of proof-technique: the identity channel has $X_{ss} = 1$ and is undetectable by any code, so some separation from the identity is necessary.
We do not improve the cardinality reached for full detection; both~\cite{BKR26} and \Cref{thm:C} reach till $\alpha < \tfrac12$.

\paragraph{A non-malleable code against global quantum tampering.}
\ifExtAbs
    Non-malleability allows the decoder to output an unrelated message, provided that this output is independent of the encoded message. Tamper detection is stronger: except with negligible probability, the decoder may only return the original message or $\perp$. Hence, \Cref{thm:A} immediately yields a non-malleable code.

    Existing quantum constructions are primarily in the split-state setting~\cite{ABW09,AM17,ABJ22,BGJR23,BB23}, while classical constructions against global tampering rely on a common reference string or computational assumptions~\cite{FMVW16,BDKLM19,DKP21}. To our knowledge, \Cref{thm:A} gives the first efficient non-malleable code against a global quantum tampering class with no split-state restriction, no shared key, and no computational hardness assumption.
\else
    A code is $\eps$-non-malleable against a family $\mathcal{F}$ if for every $f_1 \in \mathcal{F}$ there is a distribution $D_{f_1}$ over $\M \cup \{\same\}$, depending on $f_1$ and $\pp$ but not on the encoded message, such that for every $m$ the distribution of $\widehat{m}$ is within $\eps$ of $D_{f_1}$, with $\same$ read as $m$.
    Detection is the stronger requirement: it forces $D_{f_1}$ to be supported on $\{\same,\perp\}$, so the decoder never returns an unrelated message rather than merely one uncorrelated with the input.
    In other words, detection is the stronger of the two guarantees a coding scheme can offer under tampering: where non-malleability permits the decoder to return an unrelated message provided it is independent of the input, (D) forbids any output other than $m$ or $\perp$.
    \Cref{thm:A} therefore yields a non-malleable code as a corollary, and it does so in a regime the existing quantum constructions do not reach.
    The quantum non-malleability line~\cite{ABW09,AM17,ABJ22,BGJR23,BB23} works in the split-state model, where the codeword is divided into registers tampered independently, and the classical constructions against global tampering~\cite{FMVW16,BDKLM19,DKP21} rely on a common reference string or a computational assumption.
    To our knowledge the following is the first efficient non-malleable code secure against a global quantum tampering class, with no split-state restriction, no shared key, and no computational hardness assumption.
\fi

\begin{namedthm}[non-malleability]\label{cor:nm}
The scheme of \Cref{thm:A} is a non-malleable code for classical messages against every $f_1 \in \mathcal{F}_1(p)$, with error $2^{-\Omega(n)}$ and constant rate.
\end{namedthm}
\ifExtAbs
\else
    \begin{proof}
        Since tamper detection is stronger than non-malleability, the proof follows directly. See \cite{BK23} for more details.
    \end{proof}
\fi

Both clauses (D) and (ES) are useful in establishing non-malleability of the scheme.
An adversary may abort selectively on one message without ever causing a wrong decoding, so (D) alone leaves the abort probability free to depend on $m$, and no simulator blind to $m$ could then reproduce the real distribution.
What rules this out is (ES): if the accept probabilities differed noticeably between $m_0$ and $m_1$, Eve would distinguish them by watching whether Bob aborts.

\ifExtAbs
\else
    \begin{table}[ht]
    \centering
    \small
    \begin{tabular}{@{}llll@{}}
    \toprule
    Result & Adversary & Extra hypothesis & Guarantee \\
    \midrule
    \Cref{thm:A} &
    \begin{tabular}[t]{@{}l@{}}
    $\mathcal{F}_1$: circuit size $\le p$ \\
    $\mathcal{F}_2$: unbounded
    \end{tabular}
    & none & (C)+(D)+(ES), efficient \\[1.5em]
    \Cref{thm:B} & $|\Fadv| \le 2^{d^{\alpha}}$, $\alpha < 1$ & none & relaxed detection \\[0.5em]
    \Cref{thm:C} & $|\Fadv| \le 2^{d^{\alpha}}$, $\alpha < \tfrac12$ & $\Fe \le d^{-\delta}$ & full detection \\[0.5em]
    \Cref{cor:nm} & circuit size $\le p$ & none & non-malleability, efficient \\[0.5em]
    \addlinespace
    \cite[Thm.~4]{BKR26} & $|\Fadv| \le 2^{d^{\alpha}}$, $2\alpha < \delta$ & $\rank \le d^{1-\delta}$, $\Fe \le 2d^{-\delta/2}$ & full detection \\[0.5em]
    \bottomrule
    \end{tabular}
    \caption{Summary of results. Theorems~B and~C are existence statements about Haar-random unitaries; only the circuit-size results are efficient.}
    \label{tab:summary}
    \end{table}
    \paragraph{Revocable encryption.}
Consider the following situation: Alice sends a ciphertext and then changes her mind: she asks for it back, and wants a guarantee that the recipient kept nothing.
No classical scheme can give this, for the same reason that none achieves \textup{(ES)} or \textup{unclonable encryption}: the recipient copies the ciphertext, returns the original untouched, and decrypts his copy at leisure.
Quantumly the guarantee is possible using a keyed scheme, known as revocable timed-release encryption~\cite{Unr15}.

Three things are asked of such a scheme.
A recipient who keeps the ciphertext long enough can read it, so what Alice sends is genuinely useful; we call this \emph{release}, and it fixes a depth $q$ for the decoder.
A recipient who returns the ciphertext honestly passes Alice's check (\emph{revocation completeness}).
And if what he returns passes the check, then whatever he kept is independent of the message, even if he becomes computationally unbounded afterwards (\emph{revocable hiding}).

Our model supplies all three by using proof technique very similar to those described in earlier theorems.
A \emph{holder} is exactly a two-phase adversary: the phase-1 map is everything he does while the ciphertext is in his hands, and the unbounded phase-2 map is what he does after returning it, so revocable hiding is everlasting in the sense of \textup{(ES)}.
Alice's check is the decoder itself, $\Ver:=\Dec_U$, accepting when the outcome is a message and rejecting when it is $\perp$; she needs no secret and no record of which message she sent.
Revocable hiding is then the analogue of \textup{(ES)} under the map that forgets \emph{which} message the decoder returned and records only \emph{whether} it aborted, so it follows by data processing.
We state the result in terms of depth rather than size, since what Alice reasons about is elapsed time and depth is what a timing deadline bounds, with the holder's width $w$ as a separate parameter.
This is where the absence of a Kraus-rank factor is used: a holder with $w$ ancillas induces a channel of Kraus rank up to $2^{w}$, and a bound carrying $\rank(\Phi)^{\ell}$ would be vacuous for all but logarithmic $w$.

\begin{namedthm}[revocable encryption; \Cref{thm:rev}]\label{cor:rev}
Fix a polynomial $p$ and a width bound $w = \poly(n)$.
The scheme of \Cref{thm:A}, with $\Ver := \Dec_U$ accepting if and only if the outcome lies in $\M$, is a $(p,q)$-revocable keyless encryption scheme against holders of depth at most $p$ and width at most $w$, with error $2^{-\Omega(n)}$.
\end{namedthm}

\fi

\ifExtAbs
\else
    \subsection{Related work}\label{sec:related}
    
    \paragraph{Tamper detection and manipulation detection.}
    As mentioned before, in the relaxed version tamper detection asks that the decoder to either recover the message or abort.
    Classically it is unachievable against unrestricted tampering, and Jafargholi and Wichs~\cite{JW15} identified what a family must avoid for it to become possible: the constant functions, which fool the decoder into accepting a fixed codeword, and the identity, which no code can detect.
    The same two obstructions appeared in early quantum works~\cite{BK23}.
    Algebraic manipulation detection codes~\cite{CDF+08,CPX15} considers a related problem with a shared secret and an algebraic tampering class, and their quantum descendant is the Pauli manipulation detection codes of Bergamaschi~\cite{BK23,Ber24}.
    The quantum line of literature begins with unitary tampering~\cite{BK23}, continues through the split-state setting~\cite{Ber24,BB23}, and reaches general channels (CPTP maps) with Broadbent, Kapshikar and Rochette~\cite{BKR26}.
    \Cref{tab:compare} places our results among these.
    What changes with \Cref{thm:B,thm:C} is that no structural hypothesis on the tampering survives: not a basis, not a split, not a bound on Kraus rank, only cardinality and, for full detection, distance from the identity, as one expects.
    
    \begin{table}[ht]
    \centering
    \small
    \begin{tabular}{@{}llccc l@{}}
    \toprule
    & & \multicolumn{3}{c}{Restriction} & \\
    \cmidrule(lr){3-5}
    & Tampering class & size ($d=2^n$) & far from $\id$ & far from const. & Detection \\
    \midrule
    \cite{JW15} & classical functions & $2^{d^{\alpha}}$, $\alpha<1$ & required & required & strong \\
    \cite{CDF+08} & classical functions & $2^{n}$ & required & for free & strong  \\
    \cite{Ber24} & Pauli & $4^{n}$ & for free & for free & strong  \\
    \cite{BK23} & unitary & $2^{d^{\alpha}}$, $\alpha<\tfrac14$ & required & for free & strong \\
    \cite{BB23} & split-state channels$^{a}$ & N.A. & required & entangl. bound$^{b}$ & relaxed \\
    \cite{BKR26} & general channels & $2^{d^{\alpha}}$, $\alpha<\tfrac12$ & required$^{c}$ & required & strong \\
    \addlinespace
    \Cref{thm:B} & general channels & $2^{d^{\alpha}}$, $\alpha<1$ & not required & not required & relaxed \\
    \Cref{thm:C} & general channels & $2^{d^{\alpha}}$, $\alpha<\tfrac12$ & $\Fe \le d^{-\delta}$, $\delta>0$ & not required & strong \\
    \bottomrule
    \end{tabular}
    
    \smallskip
    {\footnotesize
    $^{a}$ for number of shares $t \ge 3$.\quad \\
    $^{b}$ constant channels are excluded by bounding pre-shared entanglement between shares.\quad \\
    $^{c}$ can be dropped at the cost of weakening to relaxed detection.
    }
    
    \caption{
    Tamper detection against various adversary models.}
    \label{tab:compare}
    \end{table}
    
    \paragraph{Non-malleable codes.}
    Non-malleability~\cite{DPW18} weakens detection by allowing $D_{f_1}$ to be supported on all of $\M \cup \{\same\}$ rather than on $\{\same,\perp\}$: the decoder may return an unrelated message, provided it carries no information about the one encoded.
    Capacity and rate for the classical notion are well understood~\cite{CG16}.
    The quantum constructions are almost entirely split-state~\cite{ABJ22,BGJR23,BB23}, where the codeword is divided into shares tampered locally and the separation between them substitutes for a restriction on the tampering itself; Bergamaschi and Boddu~\cite{BB23} obtain detection in this model as a consequence of non-malleability, and appear in \Cref{tab:compare} on that account.
    Keyed notions of quantum non-malleability~\cite{ABW09,AM17} are global but assume a shared secret.
    Against global tampering, the closest classical results are~\cite{CG16}, which shows existence probabilistically, and~\cite{FMVW16}, which gives an explicit construction in the common reference string model with the CRS longer than the tampering runtime.
    Our setting is of the same kind: $\pp$ is public, untamperable, and grows with $p$; and like~\cite{FMVW16} we assume no computational hardness for any problem.
    Later work removes the CRS~\cite{BDKLM19,DKP21} at the price of strong cryptographic assumptions.
    \Cref{cor:nm} is, to our knowledge, the first such explicit, efficient code against global \emph{quantum} tampering.
    \paragraph{Neighbouring primitives.}
    Quantum authentication~\cite{BCG+02,AM17} gives the same alternative, recovery or abort, against a computationally unbounded adversary, and is in that sense stronger than what we obtain; the price is a shared secret key, which is what our setting withholds.
    The comparison with quantum key distribution~\cite{BB84} is closer, since both extract an unconditional guarantee from the impossibility of undetected measurement, but the resources differ on almost every axis: QKD is interactive, needs an authenticated classical channel, and draws its security from Eve's ignorance of the preparation basis, whereas here a single message is sent, $\pp$ is public, and Eve's only restriction is phase 1 boundedness.
    
    \paragraph{Security from bounded resources.}
    Trading a bound on the adversary's resources for an unconditional guarantee is a well-established idea: the bounded-storage model~\cite{Mau92,CM97} and its quantum counterpart~\cite{DFSS05} bound memory, and Unruh~\cite{Unr13} formulates everlasting security, where an assumption in force during a protocol yields an information-theoretic conclusion afterwards.
    Clause (ES) is of this type, with online circuit size as the bounded resource.

    \paragraph{Revocable encryption.} \cite{Unr15} asks that the sender be able to demand the ciphertext back, and that a returned state passing verification leave the recipient with nothing about the message.
That construction obtains this in the random-oracle model, from a timed-release scheme as a building block.
\Cref{cor:rev} obtains the same guarantee from the resource bound alone, with the decoder itself as the verifier, and the recipient unbounded once he has returned the ciphertext.
\fi

\ifExtAbs
\else
    \subsection{Techincal Overview}\label{subsec:overview}
    We now outline the proofs of the three theorems.
    Each makes claims of two kinds about the encoded state $\ket{\psi_m} = U(\ket{m}_A \otimes \ket{0^{n-k}}_B)$.
    Define $X_{st} = \bra{\psi_t}\Phi(\proj{\psi_s})\ket{\psi_t}$ is the probability that the decoder returns $\widehat{m}=t$ when $m=s$ was encoded.
    \begin{itemize}
        \item The first is that $\sum_{t \ne s} X_{st}$ is small for every $s$ and every admissible $\Phi$, and, in the strict case, that $\sum_{t} X_{st}$ is small.
        \item The second is that for every pair $m_0, m_1 \in \M$, the state of the register retained by the adversary on the branch where the decoder does not abort is within $\eps$ in trace distance of the state for the other message, and remains so under arbitrary post-processing.
    \end{itemize}
    Both reduce to bounding moments of two scalar quantities associated with the tampering map, and this section states the required bounds, explains why the standard evaluation of such moments does not reach them, and derives the theorems. 
    \Cref{thm:A} and \Cref{thm:C} bounds on both of these scaler quantities whereas \Cref{thm:B} uses only bound on the first scaler quantity.
    
    \smallskip \noindent 
    In terms of $X_{st}$, the decoding POVM of \Cref{sec:results} is $\Pi_t = \proj{\psi_t}$ with $\Pi_\perp = \id - \sum_t \Pi_t$, so $\sum_t X_{st}$ is the probability that the decoder does not abort.
    The off-diagonal sum $\sum_{t \ne s} X_{st}$ and the diagonal term $X_{ss}$ are therefore the only quantities to be bounded, and we need to deal with them separately.
    
    \smallskip \noindent
    For $s \ne t$ the vectors $\ket{\psi_s}$ and $\ket{\psi_t}$ are orthogonal and the randomness of $U$ acts on them jointly. 
    Whereas, for $s = t$ there is a single vector and no such structure to exploit.
    Both bounds have some common structure in how these bounds are used but differ on how these bounds are established.
    We bound $\E_U[X^{\ell}]$ at a moment order $\ell$ fixed in advance, apply Markov, and take a union bound over the $2^{2k}$ ordered pairs of messages and over the adversary family.
    This union bound is what largely determines $\ell$: an adversarial family family $\Fadv$ requires $\ell$ slightly larger than $ \log_2|\Fadv|$.
    
    For \Cref{thm:A}, we have Phase 1 circuits of size\footnote{as noted before similar arguments work for depth p circuits as well} $p$, $\log_2|\mathcal{F}_1(p)| = O(p\log(n+p))$, so $\ell = \poly(n,p)$.
    For a family bounded only in cardinality, $\log_2|\Fadv| \leq  d^{\alpha}$, so $\ell \approx d^{\alpha}$, which is exponential in $n$.
    The first is compatible with an efficiently samplable encoding and the second is not, which is why \Cref{thm:A} is constructive while \Cref{thm:B} and \Cref{thm:C} are stated  only with a Haar-random $U$ (or with $t$-designs where $t=O(\ell)$ is exponential).
    Thus, all three circuits are explicit and uniform but only the first one is efficient.
    
    For a channel $\Phi$, consider $\Phi(\rho) = \sum_i K_i \rho K_i^{*}$, which gives $X_{st} = \sum_i |\bra{\psi_t}K_i\ket{\psi_s}|^{2}$.
    So $\E_U[X_{st}^{\ell}]$ is a polynomial integral over the unitary group, evaluated by Weingarten calculus.
    This is the technique used in~\cite{BK23,BKR26}.
    Expanding $X_{st}^{\ell}$ over Kraus indices leaves $r^{\ell}$ terms, where $r = \rank(\Phi)$, so the bounds carry a rank factor and survives the union bound only under a hypothesis such as $\rank(\Phi) \le d^{1-\delta}$.
    And the Weingarten coefficients are controlled only for $\ell = O(\sqrt d)$, which caps the attainable family size at $\alpha < \tfrac12$.
    
    Note that an entropic argument of the kind used in quantum key distribution is unavailable here.
    There the proof quantifies, through uncertainty relations or monogamy bounds, what the adversary loses by not knowing the preparation basis; here $\pp$ is public and there is no entropy to quantify.
    What Eve lacks is not information but the circuit to act on it, so the argument must be a moment calculation over a particular known unitary.
    
    \subsubsection*{Proof sketch for $s \neq t$}
    We give a geometric intuition of this bound here. 
    The formal bound is based on representation theory of unitary group and can be found in the techincal section.
    
    Fix $s \ne t$.
    By left-invariance of the Haar measure, $(\ket{\psi_s},\ket{\psi_t})$ is a uniformly random ordered pair of orthonormal vectors, so we may reveal $\ket{\psi_s}$ first and integrate over $\ket{\psi_t}$ afterwards.
    Once $\ket{\psi_s}$ is fixed, the operator $M = \Phi(\proj{\psi_s})$ is fixed too: it is positive, of unit trace, and otherwise arbitrary.
    Conditionally, $\ket{\psi_t}$ is uniform on the unit sphere of $\psi_s^{\perp}$, and
    \[
      \E_U\big[X_{st}^{\ell}\big]
      \;=\; \E_{\psi_s}\ \E_{\psi_t \perp \psi_s}\Big[\bra{\psi_t}M\ket{\psi_t}^{\ell}\Big].
    \]
    The inner expectation is a moment of a quadratic form in a random unit vector, which the symmetric subspace evaluates exactly.
    This gives the first estimate:
    \begin{equation*}\label{eq:ov-offdiag}
      \E_U\big[X_{st}^{\ell}\big] \;\le\; \Big(\frac{\ell}{d}\Big)^{\ell},
      \qquad s \ne t, \quad \ell < d .
    \end{equation*}
    Note that the channel was never expanded: whatever $\Phi$ does to $\ket{\psi_s}$, it produces some operator $M$.
    The property that matters for our analysis is that $M$ is a positive semidefinite matrix with trace at most 1. 
    Thus no rank factor appears, and the bound remains non-trivial for every $\ell < d$ rather than only for $\ell = O(\sqrt{d})$> So both limitations of the Weingarten route are avoided.
    Taking $\ell = \lceil d^{\alpha}\rceil$ with $\alpha < 1$ then gives \Cref{thm:B}.
    This bound is tight up to the constant in the base: for the replacement channel $\rho \mapsto \proj{\varphi}$ the inner expectation is an equality, and $\E_U[X_{st}^{\ell}] \ge (\ell/2ed)^{\ell}$.
    Hence the range $\alpha < 1$ cannot be extended by sharpening this bound.
    
    \subsubsection*{Proof sketch for s=t}
    The case for $s=t$ is more involved.
    The diagonal term $X_{ss}$ admits no such conditioning: with $s = t$ there is a single vector rather than an orthonormal pair.
    Here the channel must be expanded, but it needs to be expanded so that most of the expansion vanishes.
    Split each Kraus operator into its scalar and traceless parts,
    \[
      K_i = c_i \id + L_i, \qquad c_i = \frac{\Tr K_i}{d}, \qquad \Tr L_i = 0 .
    \]
    This decomposition then gives $X_{ss} = \Fe(\Phi) + W_s + Y_s$, where $\Fe(\Phi) = \sum_i |c_i|^{2}$ is the entanglement fidelity.
    Moreover, $\E_U[W_s] = 0$, and $Y_s = \sum_i |\bra{\psi_s}L_i\ket{\psi_s}|^{2}$.
    The first term is the part no code can detect, and the hypothesis $\Fe(\Phi) \le d^{-\delta}$ of \Cref{thm:C} is precisely the requirement that it has to be small.
    The second term has mean zero.
    So the estimate needed is only on $Y_s$.
    
    Expanding $\E_U[Y_s^{\ell}]$ over the symmetric subspace produces a sum over permutations of $2\ell$ tensor factors.
    This is where the splitting is used: every permutation with a fixed point contributes a factor $\Tr L_i = 0$, and hence does not contribute to the sum.
    The permutations that survive can be  controlled by a single scalar,
    \[
      \sum_i \big(\Tr[K_i^{*}K_i] - d|c_i|^{2}\big) = d\big(1 - \Fe(\Phi)\big).
    \]
    This gives the second bound:
    \[
      \E_U\big[Y_s^{\ell}\big] \;\le\; \Big(\frac{4\ell^{2}}{d}\Big)^{\ell} .
    \]
    Again no rank factor appears, hence we get a better handle over the overall bound than \cite{BKR26}.
    The two bounds differ in one respect that accounts for the gap between \Cref{thm:B} and \Cref{thm:C}.
    The off-diagonal bound (for $s\neq t$) is non-trivial whenever $\ell < d$, whereas the diagonal one (for $s=t$) requires $4\ell^{2} < d$, that is $\ell = O(\sqrt{d})$.
    Since the union bound forces $\ell \approx d^{\alpha}$, the first permits $\alpha < 1$ and the second only $\alpha < \tfrac12$.
    The difference is structural: $Y_s$ is quadratic in $\ket{\psi_s}$ where $X_{st}$ is bilinear in $\ket{\psi_s}$ and $\ket{\psi_t}$, so its $\ell$-th moment occupies $2\ell$ tensor factors rather than $\ell$.
    
    \subsubsection*{From bounds to \Cref*{thm:A,thm:B,thm:C} }
    \emph{Detection} follows from the first bound by a threshold argument.
    If $X_{st} \le \tau = 2^{-2k}$ for every $t \ne s$, then $\sum_{t \ne s} X_{st} \le 2^{k}\tau < 2^{-k}$, so it is enough to control each term separately.
    Markov at order $\ell$ and a union bound over the $2^{2k}$ ordered pairs and over the adversary family then close the argument, provided $\ell$ is at least the value fixed in \Cref{thm:A}, which is where the term $\log_2|\mathcal{F}_1(p)|$ in that expression shows up.
    For \Cref{thm:C} a similar argument is run again on the diagonal term, using the second bound.
    
    \smallskip \noindent
    \emph{Secrecy} requires the second bound used in a slightly different form.
    We apply second bound to the adversary's isometry slices.
    For a techincal reason, we can not directly apply these bounds to the channel but need to apply it to isometry slices. 
    Slicing Eve's Phase 1 isometry along the register $E$ (the register she retains for Phase 2) produces a family of operators whose traceless parts satisfy the same scalar bound as before. 
    Moreover, \[
      \big\| w_s - v \big\|^{2\ell} = Y_L(\psi_s),
    \]
    where $w_s$ is the vector Eve holds on the branch where Bob accepts and $v$ is a reference vector depending on her circuit but not on the message.
    In other words, the diagonal term bounds how far Eve's accept-branch state can move as the message varies.
    And since $v$ is message-independent, all $2^{k}$ of these states are in the neighbourhood of a common point (with high probability).
    The bound thus ensures that every $w_s$ is within $\eta = 2\ell\sqrt{2/d}$ of the same $v$, simultaneously for all messages and all circuits in $\mathcal{F}_1(p)$.
    Telescoping bounds can then control the accept branch, while the remaining non-abort branches carry total weight at most $\Theta(2^{-k})$ by the detection argument from above.
    Since trace distance cannot increase under the Phase 2 channel, the bound survives arbitrary unbounded post-processing, which is what \emph{everlasting} refers to.

    Both bounds are expectations of bounded-degree polynomials in the entries of $U$ and $\bar U$, so a $2\ell$-design inherits them up to a factor of at most $2$.
    Since \Cref{thm:A} uses $\ell = \poly(n,p)$, the linear-depth constructions of~\cite{MPSY24} readily give such a design by a circuit of size $\poly(n,p)$, making the scheme efficient.
\fi
    
    \bibliographystyle{bibtex/bst/alphaarxiv.bst}
    \bibliography{bibtex/bib/quasar-full.bib,
        bibtex/bib/quasar.bib,
        bibtex/bib/quasar-more-tamper-new.bib,
        bibtex/bib/quasar-more-merged.bib,
        bibtex/bib/quasar-more.bib}
\else
    
    \section{Preliminaries}

Throughout the paper all Hilbert spaces are finite-dimensional and complex. Given a Hilbert space~$\mathcal H$, we denote by $\mathbb S(\mathcal H)$ its unit sphere and by $\mathsf L(\mathcal H)$ the space of linear operators acting on $\mathcal H$. We write $\id_{\mathcal H}$ for the identity operator on $\mathcal H$, or simply $\id$ when the underlying space is clear from context. A quantum state on $\mathcal H$ is a positive semidefinite operator $\rho\in\mathsf L(\mathcal H)$ satisfying $\Tr\rho=1$, and we denote the set of quantum states on $\mathcal H$ by $\mathsf D(\mathcal H)$. For multipartite systems, subscripts indicate the registers on which a state or operator acts. For instance, $\rho_{AB}$ denotes a state on registers $A$ and $B$, while $M_A$ denotes an operator acting on register $A$.  For two quantum states $\rho,\sigma\in\mathsf D(\mathcal H)$, we denote their quantum fidelity by $F(\rho,\sigma)=\bigl\|\sqrt{\rho}\sqrt{\sigma}\bigr\|_1^2$. In particular, if $\rho=\proj{v}$ is pure, then $F(\rho,\sigma)=\bra{v}\sigma\ket{v}$.

For an operator $M$, we write $\adj{M}$ for its adjoint and $|M|=\sqrt{\adj{M}M}$ for its absolute value. For $p\in[1,\infty)$, the Schatten $p$-norm is denoted by $\|M\|_p=\bigl(\Tr |M|^p\bigr)^{1/p}$. In particular, $|\cdot|_1$ and $|\cdot|_2$ denote the trace norm and the Hilbert--Schmidt (Frobenius) norm, respectively, while $|\cdot|_\infty$ denotes the operator norm. The trace of an operator is written as $\Tr$. We use the usual Dirac notation $\ket{v}$ and $\bra{v}$ for vectors and covectors, and write $\ketbra{v}{w}$ for the rank-one operator $\ket{v}\bra{w}$. In particular, we use $\proj{v}$ for the projector onto the span of $\ket{v}$.

A quantum channel is a completely positive and trace-preserving (cptp) linear map between operator spaces. For a quantum channel $\Phi$, we denote by $\rank(\Phi)$ its Kraus rank, namely the minimum number of Kraus operators in a Kraus representation of $\Phi$, and by $\Fe(\Phi)$ its entanglement fidelity with respect to the maximally mixed state.

For a random variable $X$ distributed according to $\mu$, we write $X\leftarrow\mu$ to indicate sampling from $\mu$, and $\mathbb E_{X\sim\mu}$ for expectation with respect to this distribution. The probability of an event $E$ is denoted by $\Pr[E]$. When $\mathbb S(\mathcal H)$ is equipped with its Haar measure, the notation $v\leftarrow\mathbb S(\mathcal H)$ refers to a Haar-random unit vector.

The message set is $\mathcal M=\{0,1\}^k$, where $k$ is the message length in bits. We write $n$ for the number of qubits in the encoded state, and set $d=2^n$ for the dimension of the corresponding Hilbert space $\C^d$. The symbol $\ell\in\N$ denotes a moment order throughout the paper, its precise value will be specified when needed. We use $\negl(\cdot)$ to denote a negligible function.

\subsection{General facts}

We begin by collecting several general facts that will be useful throughout the paper.

\begin{fact}[Schatten norm inequalities {\cite{Wat18,Tom16, alma9919303308606531}}]\label{fact:schatten}
  Let $M,A_1,\dots,A_m$ be operators on $\C^{d}$ and let $p,q\in[1,\infty]$.
  Then,
  \begin{enumerate}
    \item[(i)] $|\Tr M|\le\|M\|_{1}$.
    \item[(ii)] (Hölder) If $p_1,\dots,p_m\in\N$ satisfy $\sum_{a=1}^{m}\tfrac1{p_a}=1$,
      then $\|A_1A_2\cdots A_m\|_{1}\le\prod_{a=1}^{m}\|A_a\|_{p_a}$.
    \item[(iii)] $\|M\|_{p}\le\|M\|_{q}$ whenever $q\le p$.
    \item[(iv)] $\|\adj{M}\|_{p}=\|M\|_{p}$.
  \end{enumerate}
\end{fact}

\begin{fact}[Binomial bounds]\label{fact:binom}
  For integers $1\le b\le a$,
  \[
    \left(\frac{a}{b}\right)^b\le \binom{a}{b}\le \left(\frac{ea}{b}\right)^b.
  \]
\end{fact}

\subsection{Quantum channels}

A \emph{quantum channel} (CPTP map) $\Phi:\mathsf{L}(\C^d)\to\mathsf{L}(\C^d)$ admits a Kraus representation $\Phi(\rho)=\sum_{i=1}^{r}K_i\rho \adj{K_i}$ with $\sum_i \adj{K_i} K_i=\id$. The least such $r$ is the \emph{Kraus rank} $\rank(\Phi)$. The \emph{entanglement fidelity} of $\Phi$ (relative to the maximally mixed input, following the convention of~\cite{BKR26}) is
\begin{equation}\label{eq:Fe}
  \Fe(\Phi)=\frac{1}{d^{2}}\sum_{i}\big|\Tr K_i\big|^{2},
\end{equation}
which is independent of the Kraus representation. The \emph{diamond norm} is the supremum over density operators on a doubled space,
\[
  \|\Phi-\Psi\|_\diam=\sup_{\rho}\|((\Phi-\Psi)\otimes\id)(\rho)\|_1.
\]

\subsection{The symmetric subspace}
\label{sec:symsub}

We briefly recall some basic facts about the symmetric subspace and fix the notation that will be used throughout the paper.

\begin{definition}[Symmetric subspace]\label{def:symsub}
  Let $S_N$ denote the symmetric group on $N$ letters. It acts on $(\C^{D})^{\otimes N}$ by permuting the tensor factors: for $\pi\in S_N$, the permutation operator $P_\pi$ is defined on product vectors by
  \[
    P_\pi\big(\ket{v_1}\otimes\cdots\otimes\ket{v_N}\big)=\ket{v_{\pi^{-1}(1)}}\otimes\cdots\otimes\ket{v_{\pi^{-1}(N)}},
  \]
  and extended linearly. The \emph{symmetric subspace} $\Sym^{N}(\C^{D})$ is the subspace of permutation-invariant vectors,
  \[
    \Sym^{N}(\C^{D})=\big\{\ket{\psi}\in(\C^{D})^{\otimes N}\ :\ P_\pi\ket{\psi}=\ket{\psi}\ \text{ for all }\pi\in S_N\big\}.
  \]
\end{definition}

Informally, $\Sym^{N}(\C^{D})$ is the subspace on which the ordering of the $N$ tensor factors is forgotten: a symmetric vector is determined by \emph{how many} factors carry each basis label, not by their arrangement.

\begin{lemma}[Projector onto the symmetric subspace]\label{lem:symproj}
  The orthogonal projector $\Pisym^{(N)}$ onto the symmetric subspace $\Sym^{N}(\C^{D})$ is the group average
  \[
    \Pisym^{(N)}=\frac{1}{N!}\sum_{\pi\in S_N}P_\pi .
  \]
  Moreover $\Sym^{N}(\C^{D})=\operatorname{span}\{\ket{v}^{\otimes N}:\ket{v}\in\C^{D}\}$, and $\dim(\Sym^{N}(\C^{D}))=\binom{D+N-1}{N}$.
\end{lemma}

\begin{lemma}[Permutation operators contract along cycles]\label{lem:cycletrace}
  Let $M_1,\dots,M_N$ be operators on $\C^d$ and $\pi\in S_N$, and let $P_\pi$ be as in
  \Cref{def:symsub}, acting on $(\C^d)^{\otimes N}$. Then
  \[
    \Tr\big[(M_1\otimes\cdots\otimes M_N)\,P_\pi\big]
    =\prod_{c\,\in\,\mathrm{cycles}(\pi)}
    \Tr\Big[M_{j_c}\,M_{\pi^{-1}(j_c)}\cdots M_{\pi^{-(\ell_c-1)}(j_c)}\Big],
  \]
  where $\ell_c$ is the length of the cycle $c$ and $j_c$ is any element of $c$ (the product is
  independent of this choice, by cyclicity of the trace).
\end{lemma}

For example, take $N=3$ and $\pi=(1\,2)$, which fixes $3$. Then $\pi$ has two cycles, $\{1,2\}$ and $\{3\}$, and
\[
  \Tr\big[(M_1\otimes M_2\otimes M_3)\,P_\pi\big]=\Tr[M_1M_2]\cdot\Tr[M_3].
\]
For $N=2$ the two permutations give $\Tr[(M_1\otimes M_2)P_{\mathrm{id}}]=\Tr[M_1]\Tr[M_2]$ and
$\Tr[(M_1\otimes M_2)P_{(1\,2)}]=\Tr[M_1M_2]$, the latter being the familiar swap identity.

\subsection{Representation-theoretic facts}

We next recall a few standard facts from representation theory that will be used throughout our paper.

\begin{lemma}[Schur's lemma {\cite[Lem.~1.7]{FultonHarris}}]\label{lem:schur}
  Let $(V,\pi)$ be a finite-dimensional irreducible representation of a group $G$ over $\C$. If a linear operator $T:V\to V$ satisfies $T\pi(g)=\pi(g)T$ for all $g\in G$, then $T=\lambda\,\id_V$ for some $\lambda\in\C$.\footnote{The statement in~\cite{FultonHarris} is given for finite groups, but the proof applies verbatim to any finite-dimensional representation of any group.}
\end{lemma}

\begin{fact}\label{fact:irreducible_action_on_symmetric}
  The representation $\pi$ of $\mathcal U(D)$ on $\Sym^{\ell}(\C^{D})$ defined by $\pi(V)=V^{\otimes \ell}\big|_{\Sym^{\ell}(\C^{D})}$ is irreducible.
\end{fact}

\subsection{Unitary designs}

We recall the notions of exact and approximate unitary designs that will be used throughout the paper. For a distribution $\nu$ on $\U(d)$, we denote its $t$-fold twirl by
\[
  \cM^{(t)}_\nu(X)
  :=
  \E_{U\sim\nu}
  \left[
    U^{\otimes t} X \adj{(U^{\otimes t})}
  \right].
\]
The unitary Haar $t$-fold twirl is defined analogously as
\[
  \cM^{(t)}_{\mathrm{Haar}}(X)
  :=
  \E_{U\sim\mathrm{Haar}}
  \left[
    U^{\otimes t} X \adj{(U^{\otimes t})}
  \right].
\]

\begin{definition}[Exact unitary design]\label{def:exactdesign}
  A distribution $\nu$ on $\U(d)$ is an \emph{exact $t$-design} if its $t$-fold twirl agrees with the Haar $t$-fold twirl, namely,
  \[
    \cM^{(t)}_\nu
    =
    \cM^{(t)}_{\mathrm{Haar}}.
  \]
\end{definition}

The notion of an approximate unitary design is obtained by requiring the corresponding twirling channels to be close in diamond norm.

\begin{definition}[Approximate unitary design]\label{def:approxdesign}
  A distribution $\nu$ on $\U(d)$ is an \emph{$\varepsilon$-approximate $t$-design} if
  \[
    \bigl\|
    \cM^{(t)}_\nu-\cM^{(t)}_{\mathrm{Haar}}
    \bigr\|_\diamond
    \le \varepsilon.
  \]
  In particular, for every state $\sigma$ on $(\C^d)^{\otimes t}$,
  \[
    \bigl\|
    (\cM^{(t)}_\nu-\cM^{(t)}_{\mathrm{Haar}})(\sigma)
    \bigr\|_1
    \le \varepsilon.
  \]
\end{definition}

These notions are monotone in the design order. In particular, if $\nu$ is an exact $t$-design, then it is an exact $t'$-design for every $t'\le t$. Similarly, if $\nu$ is an $\varepsilon$-approximate $t$-design, then it is an $\varepsilon$-approximate $t'$-design for every $t'\le t$.
    \section{The \texorpdfstring{$(p,q)$}{(p,q)}-tamper-detection primitive}
\label{sec:our_primitive}
Throughout, $\lambda$ is the security parameter and $k$ the message length, with $\M=\{0,1\}^{k}$.

\begin{definition}[Coding scheme with public parameters]\label{def:syntax}
  A \emph{coding scheme} $\Pi$ is a triple $(\Setup,\Enc,\Dec)$ of uniform quantum algorithms:
  \begin{itemize}
    \item $\Setup(1^\lambda,1^k)$ fixes a block length $n=n(\lambda,k)$ with $n \ge k$, and samples a \emph{public parameter} $\pp$,
    \item $\Enc(\pp,m)$, for $m\in\M$, outputs a state on a register  $C\cong\C^d$, where $d=2^n$,
    \item $\Dec(\pp,\cdot)$ is a measurement on $C\cong\C^d$ with outcomes in  $\M\cup\{\perp\}$.
  \end{itemize}
  The scheme has \emph{size} $q=q(\lambda,k)$ if $\Setup$, $\Enc$ and $\Dec$ are implementable by circuits of size $q$.
\end{definition}

A scheme $\Pi$ is \emph{perfectly complete} if $\Pr\big[\Dec(\pp,\Enc(\pp,m))=m\big]=1$ for every $\lambda,k$, every $\pp$ in the image of $\Setup(1^\lambda,1^k)$ and every $m\in\M$, it is \emph{$\mu$-complete} if the probability is at least $1-\mu(\lambda, k)$.

\begin{remark}
  No secret key appears in \Cref{def:syntax}.  The parameter $\pp$ is sampled once and is available to every party, including the adversary. In the construction developed below, $\pp$ specifies the encoding unitary $U$ of \Cref{def:enc}, security therefore does not rely on hiding the code, but on limiting the complexity of the adversary's online interaction with the codeword. The decoder, which also knows $\pp$, applies $\adj U$ and achieves perfect completeness.
\end{remark}

\subsection{Adversaries}

\begin{definition}[Two-phase adversary]\label{def:adversary}
  A \emph{$p$-bounded two-phase adversary} $\A$ is a family
  $\A=\{\A_{\lambda, k} \}_{\lambda,k\in\N}$, with $\A_{\lambda,k}=(\A^{\mathrm{msg}},V,\Lambda)$, and where:
  \begin{itemize}
    \item $\A^{\mathrm{msg}}(\pp)$ outputs a pair of messages $m_0,m_1\in\M$ together with a state on a side register $S$,
    \item On phase 1, the adversary applies an isometry $V(\pp):\C^d_C\otimes\ket0_{E_0}\to\C^d_C\otimes\Hs_E$, implemented by a quantum  circuit of size at most $p(\lambda,k)$,
    \item On phase 2, the adversary applies a channel $\Lambda(\pp):\Lop(\Hs_E)\to\Lop(\Hs_{E'})$, which is \emph{not} size-bounded.
  \end{itemize}
  Only $V$ is size-bounded, $\A^{\mathrm{msg}}$ and $\Lambda$ are unbounded.
\end{definition}

The adversary is allowed to be non-uniform. After the public parameter $\pp$ is fixed, but before the beginning of the following experiment, the adversary may perform an arbitrarily expensive computation depending on $\pp$ to determine a circuit $V \in \cG_p$, where $\cG_p$ is the family of circuits of size $p$ (see \Cref{subsec:moment} for the formal definition).
Since the circuit $V$ may be chosen as a function of the public parameter $\pp$, we require that, with high probability over the choice of $\pp$, the security guarantee hold simultaneously for every circuit $V$ of size at most $p(\lambda)$. This is stronger than requiring that, for each fixed $V$, the guarantee hold with high probability over $\pp$.

\subsection{Experiment}

We formulate security through a distinguishing experiment. A \emph{distinguisher} is an arbitrary two-outcome measurement $\{D,\id-D\}$ applied to the registers made available at the end of the experiment. Its goal is to determine which of the two challenge messages was encoded.

Fix a scheme $\Pi$, a $p$-bounded adversary $\A$, and $b\in\{0,1\}$.
Let $F$ be a classical register recording whether the decoder aborts ($\perp$), recovers the encoded message ($\same$), or decodes some other message (in which case $F$ holds that other message).
Let $P$ be a classical register holding the public parameter $\pp$.
\begin{center}
  \fbox{
    \begin{minipage}{0.92\textwidth}
      $\Expt^{\NTDshort}_{\Pi,\A}(\lambda,k,b)$:
    \begin{enumerate}\itemsep2pt
        \item $\pp\leftarrow\Setup(1^\lambda,1^k)$.
        \item $(m_0,m_1,\rho_S)\leftarrow\A^{\mathrm{msg}}(\pp)$.
        \item $\rho_C\leftarrow\Enc(\pp,m_b)$.
        \item $\rho_{CE}\leftarrow V \rho_C V^{*}$.
        \item $\hat m\leftarrow\Dec(\pp,\rho_{C})$, set $f=\perp$ if $\hat m=\perp$,
          $f=\same$ if $\hat m=m_b$, and $f=\hat m$ otherwise.
        \item Output the state $\big(\pp,\ f,\ \Lambda(\rho_E),\ \rho_S\big)$ on  $P\otimes F\otimes E' \otimes S$.
      \end{enumerate}
  \end{minipage}}
\end{center}
\noindent
We use the shorthand $\Expt^{\NTDshort}_{\Pi,\A}(\lambda,k,b)$ to denote the state $\big(\pp,\ f,\ \Lambda(\rho_E),\ \rho_S\big)$ it outputs.

The public parameter is released to the distinguisher because it is public: it is handed to every party, the adversary included, and withholding it in the analysis would make the security notion weaker than the setting it models.
Formally, since $P$ is classical, the trace distance between the two challenge output states equals the average over $\pp$ of the trace distance of the conditional states, which is the quantity appearing in \Cref{def:sim} below.

In the security experiment, the distinguisher never learns the decoded message on the branch where decoding succeeds: handing it $\hat m$ there would reveal $b$ outright, and no scheme could then satisfy \textup{(ES), see \Cref{def:pq}}.
On the remaining branches, $f$ is fully informative: it is $\perp$ on abort, and otherwise the decoded message itself.
Releasing $\hat m$ there only strengthens the adversary, so the definition loses nothing by including it, the case is retained so that the experiment is defined for every scheme, and by clause \textup{(D)} of \Cref{def:pq} it occurs with probability $\negl(\lambda)$.
Thus, the coarsening is asymmetric, but in neither direction does it weaken the
adversary: on the success branch, withholding $\hat m$ is faithful, since a real adversary observes only that the transmission was accepted and never sees $\hat{m}$ itself, elsewhere, the distinguisher is given more than it could obtain in reality.

Note that the variable $\same$ is not actually computed by any of the parties involved in the protocol: the decoder learns $\hat m$ but not $m_b$, and the adversary learns neither, so the test $\hat m = m_b$ is available only in the analysis.
What \emph{is} available to the adversary is the abort bit, since it observes whether the transmission was accepted.
This is why the accept probability must itself be message-independent up to $\negl(\lambda)$.
Clause \textup{(ES)} below imposes this implicitely, without a separate requirement.
The symbol $\same$ follows the convention in the non-malleable-code experiments of \cite{DPW18,FMVW16}.

\paragraph{Branch decomposition.}
Before the phase-2 channel $\Lambda$ is applied, the decoder's outcomes are mutually orthogonal. Thus, after encoding, phase-1 isometry $V$, and decoding, the joint state $\sigma_{\pp,V}(m)$ on $F\otimes E$ is a classical--quantum (cq) state of the form:
\begin{equation}\label{eq:branch-decomp}
  \sigma_{\pp,V}(m)\;=\;\sum_{f}\proj{f}_F\otimes\sigma_f(m),
\end{equation}
where each $\sigma_f(m)$ is a subnormalized state on $E$ and $\Tr\sigma_f(m)$ is the probability that the flag takes the value $f$ when $m$ is encoded.
We write
\begin{equation}\label{eq:nonabort}
  \tilde\sigma_{\pp,V}(m)\;=\;\sum_{f\ne\perp}\proj{f}_F\otimes\sigma_f(m)
\end{equation}
for the restriction of \cref{eq:branch-decomp} to the branches on which the decoder does \emph{not} abort.

\begin{definition}[$(p,q)$-\NTD]\label{def:pq}
  A coding scheme $\Pi$ of size $q(\lambda,k)$ is a \emph{$(p,q)$-\NTD\ scheme} if the following three conditions hold.

  \medskip\noindent\textbf{(C) Completeness.} Untampered codewords decode
  correctly:
  \[
    \Pr\big[\Dec(\pp,\Enc(\pp,m))=m\big]=1
  \]
  for every $\lambda,k$, every $\pp$ in the image of $\Setup(1^\lambda,1^k)$ and
  every $m\in\M$.

  \medskip
  The remaining two conditions \textup{(D)} and \textup{(ES)} are required to hold for every $p$-bounded two-phase adversary $\A$.

  \medskip\noindent\textbf{(D) Detection.} For every $b\in\{0,1\}$, the decoder either aborts or recovers the encoded message with high probability:
  \[
    \Pr\Big[\hat m\notin\{\perp,\,m_b\}\ \text{ in }\
    \Expt^{\NTDshort}_{\Pi,\A}(\lambda,k,b)\Big]\;\le\;\negl(\lambda).
  \]
  Or equivalently, $\Pr\big[f\notin\{\perp,\same\}\big]\le\negl(\lambda)$.

  \medskip\noindent\textbf{(ES) Everlasting secrecy.} With $\tilde\sigma_{\pp,V}$ as in \cref{eq:nonabort}, for every phase-2 channel $\Lambda$,
  \[
    \E_{\pp}\Big\|\big(\id_F\otimes\Lambda\big)\big(\tilde\sigma_{\pp,V}(m_0)\big)
    -\big(\id_F\otimes\Lambda\big)\big(\tilde\sigma_{\pp,V}(m_1)\big)\Big\|_1
    \;\le\;\negl(\lambda).
  \]
  Or equivalently, by the operational characterization of the trace distance, for every distinguisher $\Dist$ that is given the experiment's output when $f\ne\perp$, (without renormalization),
  \[
    \big|\Pr[\Dist(\Expt_0)=1]-\Pr[\Dist(\Expt_1)=1]\big|
    \;\le\;\negl(\lambda).
  \]
\end{definition}

The two security requirements \textup{(D)} and \textup{(ES)} capture complementary aspects of tamper detection. Condition \textup{(D)} guarantees integrity: except with negligible probability, decoding either rejects or returns the message that was originally encoded. Condition \textup{(ES)} provides secrecy on the non-abort branches: even after arbitrary post-processing of the adversary's retained register, these branches are information-theoretically indistinguishable for the two challenge messages.

\subsection{A simulation-based equivalent of (ES)}

\begin{definition}[Simulation security]\label{def:sim}
  A scheme $\Pi$ is \emph{simulation-secure against $p$-bounded adversaries} if for every $p$-bounded adversary $\A$ there is a family of subnormalized states $\{\tilde\sigma^\star_{\pp,V}\}$ on $F\otimes E$, depending only on $\pp$ and $V$ (and crucially \emph{not} on the message), such that for every $m\in\M$ and every phase-2 channel $\Lambda$,
  \[
    \E_{\pp}\Big\|\big(\id_F\otimes\Lambda\big)\big(\tilde\sigma_{\pp,V}(m)\big)
    -\big(\id_F\otimes\Lambda\big)\big(\tilde\sigma^\star_{\pp,V}\big)\Big\|_1
    \;\le\;\negl(\lambda).
  \]
\end{definition}

\noindent
The channel $\Lambda$ is formally redundant here: the trace distance cannot increase under $\id_F\otimes\Lambda$.
We display it to make the two-stage guarantee explicit: indistinguishability thus holds not merely at the end of phase 1, but after arbitrary unbounded post-processing of the retained register.

\begin{proposition}\label{prop:ind-sim}
  $\Pi$ satisfies \textup{(ES)} of \Cref{def:pq} if and only if it is
  simulation-secure in the sense of \Cref{def:sim}.
\end{proposition}
\begin{proof}
  Suppose first that $\Pi$ is simulation-secure. Then, for any
  $m_0,m_1\in\M$, the triangle inequality gives
  \begin{equation*}
    \begin{aligned}
      \Big\|
      (\id_F\otimes\Lambda)\big(\tilde\sigma_{\pp,V}(m_0)\big)
      -
      (\id_F\otimes\Lambda)\big(\tilde\sigma_{\pp,V}(m_1)\big)
      \Big\|_1
      &\le
      \Big\|
      (\id_F\otimes\Lambda)\big(\tilde\sigma_{\pp,V}(m_0)\big)
      -
      (\id_F\otimes\Lambda)\big(\tilde\sigma^\star_{\pp,V}\big)
      \Big\|_1 \\
      &\quad+
      \Big\|
      (\id_F\otimes\Lambda)\big(\tilde\sigma^\star_{\pp,V}\big)
      -
      (\id_F\otimes\Lambda)\big(\tilde\sigma_{\pp,V}(m_1)\big)
      \Big\|_1,
    \end{aligned}
  \end{equation*}
  which is negligible by \Cref{def:sim}. Hence by taking the expectation over $\pp$, condition \textup{(ES)} holds.

  Conversely, assume \textup{(ES)}, fix any reference message $m^\star\in\M$, and define $\tilde\sigma^\star_{\pp,V} =\tilde\sigma_{\pp,V}(m^\star)$. Applying \textup{(ES)} to the pair $(m,m^\star)$ gives precisely the
  bound of \Cref{def:sim} for every $m\in\M$.
\end{proof}

\subsection{Existence of \texorpdfstring{$(p,q)$}{(p,q)}-tamper-detection schemes}

The main result of this paper shows that tamper detection can be achieved against non-uniform adversaries of any prescribed polynomial size, while retaining everlasting secrecy after the tampering phase.

\begin{theorem}[$(p,q)$-tamper detection]\label{thm:pq-main}
  For every polynomial $p(\lambda,k)$, there exists a larger polynomial $q(\lambda,k)$ such that a perfectly complete $(p,q)$-\NTD\ scheme exists and is efficiently implementable with approximately unitary design.
\end{theorem}

We conclude this section with a simple lower bound showing that the expansion factor required by our construction is essentially optimal.

\begin{proposition}\label{prop:expansion}
Let $(\enc,\dec)$ encode $k$-bit messages into $n$ qubits, $d=2^{n}$, with perfect completeness and decoding POVM $\{\Pi_t\}_{t\in\M}\cup\{\Pi_\perp\}$.
Then there is a stabiliser state $\ket\varphi$ and a message $m$ such that the constant channel $\Phi_\varphi:\rho\mapsto\proj{\varphi}$ satisfies
\[
  \Pr\big[\dec(\Phi_\varphi(\enc(m)))\notin\{\perp,m\}\big]\ \ge\ \big(1-2^{-k}\big)\frac{2^{k}}{d}.
\]
Consequently any $\eps$-relaxed tamper detecting scheme against a family containing the constant channels has $\eps\ge 2^{-(\gamma-1)k-1}$, so $\eps=2^{-k}$ forces $\gamma\ge 2-1/k$.
Moreover $\Phi_\varphi$ has circuit size $O(n^{2})$, so the bound applies to $(p,q)$-tamper detection for every $p=\Omega(n^{2})$.
\end{proposition}
\begin{proof}
Perfect completeness gives $\Tr[\Pi_t\enc(t)]=1$, and $0\preceq\Pi_t\preceq\id$ then forces $\Pi_t$ to act as the identity on the support of $\enc(t)$, so $\Tr\Pi_t\ge1$ and $\Tr\Pi\ge2^{k}$ for $\Pi=\sum_t\Pi_t$.
Stabiliser states form a $1$-design, so $\E_\varphi\bra{\varphi}\Pi\ket{\varphi}=\Tr\Pi/d\ge2^{k}/d$ and some stabiliser $\varphi$ attains the mean.
As $\Phi_\varphi$ discards its input, the output distribution is independent of the encoded message, so averaging over uniform $m$ gives $\E_m\big[\sum_{t\ne m}\bra{\varphi}\Pi_t\ket{\varphi}\big]=(1-2^{-k})\bra{\varphi}\Pi\ket{\varphi}$, and some $m$ attains the average.
With $d=2^{\gamma k}$ this is at least $2^{-(\gamma-1)k-1}$.
Finally a stabiliser state is prepared from $\ket{0^{n}}$ by $O(n^{2})$ Clifford gates, and Eve's precomputation is free, so she identifies $\varphi$ offline and spends only those gates in the window.
\end{proof}

The proof of \Cref{thm:pq-main} is given in \Cref{sec:proof} below, the scheme is efficient, the polynomial $q(\lambda,k)$ can be computed explicitly given $p(\lambda,k)$. The proof proceeds in the following steps. We first consider an idealized construction in which the encoding is obtained from a Haar-random unitary, and choose a sufficiently large moment order to obtain guarantees simultaneously against all phase 1 adversaries of the prescribed circuit size. We then establish the moment bounds needed to control the relevant overlaps of the encoded states under tampering. These bounds imply, with overwhelming probability over the choice of the encoding unitary, both tamper detection and everlasting secrecy, while perfect completeness follows directly from the construction. Finally, we replace the Haar-random unitary by an approximate unitary design of sufficiently high order, showing that the same guarantees are preserved and yielding an efficient implementation of the scheme. 
    \section{Proof of \Cref*{thm:pq-main}} \label{sec:proof}

\subsection{Construction and choice of parameters}

We first describe the construction and fix the parameters that will be used throughout the proof.  We then establish the moment estimates needed to prove, in order, perfect completeness, detection, and everlasting secrecy.

\subsubsection{The scheme}

We begin with an idealized Haar-random construction.  The efficient implementation of this construction, obtained by replacing the Haar unitary by an approximate unitary design, is given later in \Cref{sec:efficient-pq}.

We recall the construction and definitions from \cite{BKR26}. Fix the message set $\mathcal M=\{0,1\}^k$. Split the $n$ code qubits in two registers: register $A$ (the $k$ message qubits) and register $B$ (the remaining $n-k$ qubits). The \emph{expansion factor} is $\gamma=n/k$.

\begin{definition}[Haar encoding and projector decoder]\label{def:enc}
  Let $U\in U(d)$.  The public parameter is $\pp=U$.
  For every message $m\in\M$, define
  \[
    \ket{\psi_m}
    =
    U\big(\ket{m}_A\otimes\ket{0}_B\big).
  \]
  The encoder is $\Enc_U(m)=\ket{\psi_m}$.

  The decoder $\Dec_U$ applies $\adj U$, measures $B$ in the computational basis, outputs $\perp$ if the outcome is nonzero, and otherwise measures $A$ in the computational basis and outputs the resulting message.

  Equivalently, $\Dec_U$ is the projective measurement $\{\Pi_m\}_{m\in\M}\cup\{\Pi_\perp\}$, where $\Pi_m=\proj{\psi_m}$ and $\Pi_\perp=\id-\sum_{m\in\M}\Pi_m$.
\end{definition}

Thus the idealized scheme is completely specified by the choice of the unitary $U$.  We will show that, with overwhelming probability over a Haar-random choice of $U$, this scheme simultaneously satisfies the required security properties.

\subsubsection{The adversary class and moment order} \label{subsec:moment}

We now fix the finite class of bounded phase 1 adversary's isometries. To obtain a statement that holds simultaneously for all of them, we first bound the size of the relevant circuit class and choose the moment order accordingly, this will also determine the design order required below.

We choose $n=n(\lambda,k)$ to be polynomial in $(\lambda,k)$ and sufficiently large. Moreover, since $p=p(\lambda,k)$ is a fixed polynomial and $\lambda,k\le n$, we have $p(\lambda,k)=\poly(n)$.
We fix a constant $c$ and a finite universal gate set $\G$ acting on at most $c$ qubits at a time, and for a polynomial $p=p(\lambda,k)$ we write
\begin{equation*}
  \cG_p \;=\;
  \Big\{\, V:\C^d_C\otimes\ket{0}_{E_0}\to\C^d_{C}\otimes\Hs_E
  \;:\; V \text{ is implemented by a $\G$-circuit of size at most } p \,\Big\}.
\end{equation*}
The set $\cG_p$ is finite because $\G$ is, and it is fixed before $U$ is drawn, however, the choice of $V\in\cG_p$ may depend on $U$. A bound holding for all $V\in\cG_p$ simultaneously therefore also covers attacks that depend on $U$.
\begin{lemma}[Gate counting]\label{fact:gatecount}
  The number of phase-1 circuits of size at most $p$ satisfies
  \[
    \log_2|\cG_p|=\poly(n).
  \]
\end{lemma}
\begin{proof}
  A circuit of size at most $p$ acts on the $n$ code qubits and on at most $cp$ additional qubits, since each gate acts on at most $c$ qubits. Thus every such circuit acts on at most
  \[
    N=n+cp
  \]
  qubits.
  For a circuit with exactly $j$ gates, each gate is specified by its type, with at most $|\G|$ choices, and by the qubits on which it acts, with at most $N^c$ choices. Hence the number of circuits with exactly $j$ gates is at most $(|\G|N^c)^j$. Summing over $0\le j\le p$, the number of circuits of size at most $p$ is at most
  \[
    \sum_{j=0}^{p}\bigl(|\G|N^c\bigr)^j
    \le
    (p+1)\bigl(|\G|N^c\bigr)^p.
  \]
  Since different circuit descriptions may implement the same isometry, this also upper bounds $|\cG_p|$. Therefore,
  \[
    \log_2|\cG_p|
    \le
    \log_2(p+1)
    +p\log_2|\G|
    +cp\log_2(n+cp)
    =
    \poly(n).
  \]
\end{proof}

For a circuit of a $V \in \cG_p$, the register $E_0$ has $O(cp)$ qubits, so the channel $\Phi_V$ induced by $V$ may have Kraus rank $2^{\Theta(p)}$. Although the phase 1 adversary's map may have an arbitrary number of Kraus operators, we will not need to keep track of their number individually. Their contributions will always be combined using the completeness relation $\sum_e\adj{A_e}A_e=\id$ rather than term by term.

Throughout the proof, we use the moment order
\begin{equation}\label{def:momentorder}
  \ell
  =
  \Big\lceil
  k+\log_2|\cG_p|+n
  \Big\rceil .
\end{equation}

By \Cref{fact:gatecount}, we have $\ell=\poly(n)$, and hence $2\ell<d=2^n$ for all sufficiently large $n$. This choice of $\ell$ will serve as the moment order of the unitary $2\ell$-design used in the construction and, in turn, determine the polynomial $q$ appearing in \Cref{thm:pq-main}.

\subsection{Technical lemmas}

\subsubsection{Overlap variables}

\begin{definition}[Overlap variables]\label{def:overlap}
  For a channel $\Phi$ with Kraus operators $\{K_i\}$, two messages $s,t\in\mathcal M$ and a unitary $U\in U(d)$, we define the overlap variables $X_{st}$ as
  \[
    X_{st}=\bra{\psi_t}\Phi(\proj{\psi_s})\ket{\psi_t}=\sum_{i}\big|\bra{\psi_t}K_i\ket{\psi_s}\big|^{2}.
  \]
  Then $X_{st}$ is exactly the probability that $\Dec_U$ outputs $t$ when $m=s$ was encoded and $\Phi$ applied. It follows that $\sum_{t\in\mathcal M}X_{st}\le1$, with $\Pr[\Dec_U=\perp\mid s]=1-\sum_{t}X_{st}$. This probability is independent of the choice of Kraus representation.
\end{definition}

The following lemma shows that the diagonal overlap $X_{ss}$ is irrelevant for relaxed tamper detection, so we may focus on the off-diagonal overlaps $X_{st}$ for $s\ne t$.

\begin{lemma}[Outcome-partition identity]\label{lem:dropdiag}
  For every $s\in\mathcal M$, every channel $\Phi$, and every $U$,
  \[
    \Pr\big[\Dec_U(\Phi(\Enc_U(s)))\in\{\perp,s\}\big]=1-\sum_{t\ne s}X_{st}.
  \]
  Consequently the code is $\varepsilon$-relaxed-tamper-detecting against $\Fadv$ if and only if $\sum_{t\ne s}X_{st}\le\varepsilon$ for all $s\in\mathcal M$ and $\Phi\in\Fadv$.
\end{lemma}
\begin{proof}
  The decoder's outcomes partition the unit of probability as
  $1=\Pr[\perp]+\sum_{t\in\mathcal M}\Pr[\text{output }t]=\big(1-\sum_{t}X_{st}\big)+\sum_{t}X_{st}$.
  The relaxed-good event is $\{\perp\}\cup\{\text{output }s\}$, whose probability is
  \[
    \Pr[\perp]+\Pr[\text{output }s]=\Big(1-\sum_{t}X_{st}\Big)+X_{ss}=1-\sum_{t\ne s}X_{st}.
  \]
\end{proof}

\subsubsection{Haar moment estimates}

We first establish some Haar moment properties that will be used in the proof of \Cref{thm:pq-main}.

\begin{lemma}[Symmetric-subspace moments]\label{lem:sym-moment}
  Let $u \leftarrow \mathbb{S}(\C^{D}) $ be Haar-distributed on the unit sphere of $\C^{D}$ and let $\ell\ge1$ be an integer. Then we have the following two Haar moment properties:
  \begin{enumerate}
    \item[(i)] $\displaystyle \E_u\big[(\proj{u})^{\otimes\ell}\big]=\frac{\Pisym^{(\ell)}}{\binom{D+\ell-1}{\ell}}$,
    \item[(ii)] for any operator $B\succeq0$ on $\C^{D}$, $\displaystyle \E_u\big[\bra{u}B\ket{u}^{\ell}\big]\le\frac{(\Tr B)^{\ell}}{\binom{D+\ell-1}{\ell}}$.
  \end{enumerate}
\end{lemma}
\begin{proof}
  Let $\sigma=\E_u\big[(\proj{u})^{\otimes \ell}\big]$ regarded as an operator on $(\C^{D})^{\otimes \ell}$. Since
  $\ket{u}^{\otimes \ell}\in\Sym^{\ell}(\C^{D})$, we have
  \[
    (\proj{u})^{\otimes\ell}
    =\Pisym^{(\ell)}(\proj{u})^{\otimes\ell}\Pisym^{(\ell)}
  \]
  for every $u$, and hence $\sigma=\Pisym^{(\ell)}\sigma\Pisym^{(\ell)}$. Thus $\sigma$ vanishes on $\Sym^{\ell}(\C^{D})^{\perp}$, and we may take its restriction to $\Sym^{\ell}(\C^{D})$ as an operator on that subspace. For every $V\in\mathcal U(D)$,
  \[
    V^{\otimes \ell}\,\sigma\,\adj{(V^{\otimes \ell})}=\E_u\Big[\big(V\proj{u}\adj{V}\big)^{\otimes \ell}\Big]=\E_u\Big[(\proj{Vu})^{\otimes \ell}\Big]=\sigma,
  \]
  the last equality because the Haar measure is invariant under $u\mapsto Vu$. Thus $\sigma$ commutes with the irreducible representation of \Cref{fact:irreducible_action_on_symmetric}, so by Schur's lemma $\sigma=c\,\Pisym^{(\ell)}$; since $\Tr\sigma=1$ we get $c=\binom{D+\ell-1}{\ell}^{-1}$. Now
  \begin{align*}
    \E_u\big[\bra u B\ket u^{\ell}\big]
    &=\E_u\Tr\big[B^{\otimes \ell}(\proj u)^{\otimes \ell}\big]\\
    &=\Tr\big[B^{\otimes \ell}\sigma\big] &&\mbox{by linearity}\\
    &=c\,\Tr\big[B^{\otimes \ell}\Pisym^{(\ell)}\big] &&\mbox{since }\sigma=c\,\Pisym^{(\ell)}\\
    &\le c\,\Tr\big[B^{\otimes \ell}\big] &&\mbox{since }B^{\otimes \ell}\succeq0\mbox{ and }\Pisym^{(\ell)}\preceq\id\\
    &=\frac{(\Tr B)^{\ell}}{\binom{D+\ell-1}{\ell}} &&\mbox{since }\Tr\big[B^{\otimes \ell}\big]=(\Tr B)^{\ell}.
  \end{align*}
\end{proof}

Recall that, for distinct messages $s \neq t$, the overlap variable $X_{st}$ denotes the random variable measuring the corresponding off-diagonal contribution to the tamper-detection experiment, as defined in \Cref{def:overlap}.

\begin{lemma}[Rank-free distinct-message moments]\label{lem:rankfree}
  Let $\Phi:\mathsf{L}(\C^d)\to\mathsf{L}(\C^d)$ be any quantum channel, let $U$ be a Haar-random unitary on $\C^d$, and let $s\ne t$ be distinct messages. Then for every integer $\ell\ge1$,
  \[
    \E_U\big[(X_{st})^{\ell}\big]\le\binom{d+\ell-2}{\ell}^{-1}.
  \]
  And for $\ell\ge2$ we have the simpler bound $\E_U[(X_{st})^{\ell}]\le\big(\tfrac{\ell}{d}\big)^{\ell}$.
\end{lemma}
\begin{proof}
  Write the Haar-random unitary as acting on the fixed orthonormal vectors $\ket s,\ket t$ (here we abuse notation, identifying $\ket m$ with $\ket m_A\otimes\ket0_B$). Set $\ket{\psi_s}=U\ket s$ and $\ket{\psi_t}=U\ket t$. We condition on $\ket{\psi_s}$ and integrate over $\ket{\psi_t}$.

  \emph{Conditional distribution of $\ket{\psi_t}$.} Since $\ket s\perp\ket t$, by the left invariance of the Haar measure, the pair $(\ket{\psi_s},\ket{\psi_t})$ is distributed as a uniformly random ordered pair of orthonormal vectors.
  Conditioned on the value $v=\ket{\psi_s}$, the vector $u=\ket{\psi_t}$ is Haar-distributed on the unit sphere of the orthogonal complement $v^{\perp}\cong\C^{d-1}$.

  \emph{Reduction to a positive form.} Let $A=\Phi(\proj{\psi_s})=\Phi(\proj v)$, a density operator: $A\succeq0$, $\Tr A=1$. Let $P=\id-\proj v$ be the projector onto $v^{\perp}$ and define $B=PAP$, an operator supported on $v^{\perp}$ with $B\succeq0$ and
  \[
    \Tr B=\Tr[PAP]=\Tr[AP]=1-\bra v A\ket v\le1.
  \]
  For $u\in v^{\perp}$ we have $\bra u A\ket u=\bra u PAP\ket u=\bra u B\ket u$, hence $X_{st}=\bra{\psi_t}A\ket{\psi_t}=\bra u B\ket u$.

  \emph{Symmetric-subspace integral.} Apply \Cref{lem:sym-moment} with $D=d-1$ (the dimension of $v^{\perp}$), to the operator $B$ restricted to $v^{\perp}$:
  \[
    \E_u\big[\bra u B\ket u^{\ell}\big] \le \frac{(\Tr B)^{\ell}}{\binom{d+\ell-2}{\ell}} \leq \frac{1}{\binom{d+\ell-2}{\ell}} \ .
  \]
  Since this bound is independent of the conditioned $v$, taking the expectation over $v$ preserves it.

  The simpler bound $\E_U[(X_{st})^{\ell}]\le\big(\tfrac{\ell}{d}\big)^{\ell}$ for $\ell\ge2$ follows directly from the binomial bound of \Cref{fact:binom}.
\end{proof}

Two properties of \Cref{lem:rankfree} drive results below.
First, \emph{rank independence}: a single-moment Markov inequality gives $\Pr[X_{st}>\tau]\le(n/(d\tau))^{n}$, which can be made to beat a family of cardinality $2^{d^\alpha}$ by taking the moment order $n\sim d^\alpha$, with no constraint on the Kraus rank.
Second, \emph{validity for all $n<d$}: the estimate is exact for every moment order below the dimension, so $n\sim d^\alpha$ may be taken with $\alpha$ arbitrarily close to $1$. Moment estimates obtained by Weingarten calculus instead carry a multiplicative factor $r^{n}$ in the Kraus rank and hold only for $n=\mathcal{O}(\sqrt d)$: the rank factor forces a rank restriction on the adversarial family, and the validity ceiling caps the attainable family size at $\alpha<\tfrac12$. The conditioning argument incurs no cost.

\subsubsection{Structure of the permutation expansion}

We first record the structure of the permutation expansion and isolate two properties that will be used repeatedly in the moment estimates below.

\begin{definition}[Traceless family and its weight]\label{def:traceless-family}
  Let $\mathcal L=\{L_i\}_{i=1}^{r}$ be a finite family of operators on $\C^d$ such that $\Tr L_i=0$ for every $i$.
  Define
  \[
    Y_{\mathcal L}(v)=\sum_i\big|\bra v L_i\ket v\big|^2 \qquad \text{and} \qquad \Sigma_{\mathcal L}=\sum_i \adj{L_i}L_i.
  \]
\end{definition}

To analyze the moments of $Y_{\mathcal L}$, we first write
\[
  Y_{\mathcal L}(v)  = \sum_i\big|\bra v L_i\ket v\big|^2
  = \sum_i \bra v L_i\ket v \overline{\bra v L_i\ket v}
  = \sum_i \Tr \left(L_i \ketbra{v}{v} \right) \Tr \left( \adj{L_i}  \ketbra{v}{v}\right).
\]
Raising to the $\ell$-th power and expanding over the indices gives
\begin{equation} \label{eq:Ypower}
  Y_{\mathcal L}^{\ell}(v)
  = \sum_{i_1,\dots,i_\ell \in [r]}
  \prod_{a=1}^{\ell}
  \Tr[L_{i_a}\ketbra{v}{v}]\,
  \Tr[\adj{L_{i_a}}\ketbra{v}{v}]
  = \sum_{i \in [r]^\ell}
  \Tr\big[\mathbf M(i)\,\ketbra{v}{v}^{\otimes 2\ell}\big],
\end{equation}
where we used
\[
  \prod_{j=1}^{N}\Tr[M_j\ketbra{v}{v}]
  = \Tr\big[(M_1\otimes\cdots\otimes M_N)
  \ketbra{v}{v}^{\otimes N}\big].
\]
Here $N=2\ell$, and for $i=(i_1,\dots,i_\ell)\in[r]^\ell$ we set
\begin{equation}\label{eq:slots}
  \mathbf M(i)
  = \bigotimes_{a=1}^{\ell}
  \big(L_{i_a}\otimes \adj{L_{i_a}}\big)
  = M_1\otimes\cdots\otimes M_N,
  \qquad
  M_{2a-1}=L_{i_a},\quad
  M_{2a}=\adj{L_{i_a}}.
\end{equation}
Thus, each index $i_a$ is associated with a distinguished pair of tensor
slots, namely $(2a-1,2a)$: the first contains $L_{i_a}$ and the second
contains its adjoint $\adj{L_{i_a}}$.

We first make explicit the permutation expansion that will serve as the starting point for the analysis below.

\begin{lemma}[Permutation expansion]\label{prop:permexp}
  Let $\mathcal L$ be a finite family of operators and let $\ell\ge1$. Then
  \[
    \E_{v \leftarrow \mathbb{S}(\C^{d})}\big[Y_{\mathcal L}^{\ell}(v)\big]
    =\frac{1}{(2\ell)!\,\binom{d+2\ell-1}{2\ell}}\sum_{\pi\in S_{2\ell}}\val(\pi),
  \]
  where $\val(\pi)=\sum_{i \in [r]^\ell}\Tr\big[\mathbf M(i)\,P_\pi\big]$ and  $\mathbf M(i)$ is as in eq.~\cref{eq:slots}.
\end{lemma}
\begin{proof}
  Write $\rho=\ketbra{v}{v}$ and $N=2\ell$. By the expansion \cref{eq:Ypower} we have,
  \[
    Y_{\mathcal L}^{\ell}(v)=\sum_{i\in[r]^{\ell}}\Tr\big[\mathbf M(i)\,\rho^{\otimes N}\big],
  \]
  so that all dependence on $v$ is confined to the single factor $\rho^{\otimes N}$, which enters linearly.
  Hence,
  \begin{align*}
    \E_v\big[Y_{\mathcal L}^{\ell}(v)\big]
    &=\E_v\bigg[\sum_{i\in[r]^{\ell}}\Tr\big[\mathbf M(i)\,\rho^{\otimes N}\big]\bigg]\\
    &=\sum_{i\in[r]^{\ell}}\Tr\Big[\mathbf M(i)\,\E_v\big[\rho^{\otimes N}\big]\Big]
    &&\mbox{by linearity}\\
    &=\sum_{i\in[r]^{\ell}}\Tr\bigg[\mathbf M(i)\,\frac{\Pisym^{(N)}}{\binom{d+N-1}{N}}\bigg]
    &&\mbox{by \Cref{lem:sym-moment}}\\
    &=\frac{1}{\binom{d+N-1}{N}}\sum_{i\in[r]^{\ell}}\Tr\bigg[\mathbf M(i)\,\frac{1}{N!}\sum_{\pi\in S_{N}}P_\pi\bigg]
    &&\mbox{by \Cref{lem:symproj}}\\
    &=\frac{1}{N!\,\binom{d+N-1}{N}}\sum_{\pi\in S_{N}}\ \sum_{i\in[r]^{\ell}}\Tr\big[\mathbf M(i)\,P_\pi\big]
    &&\\
    &=\frac{1}{(2\ell)!\,\binom{d+2\ell-1}{2\ell}}\sum_{\pi\in S_{2\ell}}\val(\pi).
  \end{align*}
\end{proof}

\begin{example}[Case $\ell=1$]\label{ex:ell1}
  Let $\mathcal L=\{L_i\}_{i=1}^{r}$ be a finite family of traceless operators on $\C^d$.
  Here $N=2$, $\mathbf M(i)=L_i\otimes\adj{L_i}$, and $S_2=\{\mathrm{id},(1\,2)\}$. We get,
  \[
    \val(\mathrm{id})=\sum_i\Tr[L_i]\,\Tr[\adj{L_i}]=0,
    \qquad
    \val\big((1\,2)\big)=\sum_i\Tr\big[L_i\adj{L_i}\big]=\Tr\Sigma_{\mathcal L},
  \]
  the first vanishing because each $L_i$ is traceless. Since $2!\binom{d+1}{2}=d(d+1)$,
  \[
    \E_v\big[Y_{\mathcal L}(v)\big]=\frac{\Tr\Sigma_{\mathcal L}}{d(d+1)} .
  \]
\end{example}

Thus, already at $\ell=1$ the identity permutation is annihilated by tracelessness, and the single surviving permutation is the pairing.
We will generalize that phenomenon in the following lemma.

\begin{lemma}\label{lem:C}
  Let $\mathcal L=\{L_i\}_{i=1}^{r}$ be a finite family of traceless operators on $\C^d$. Let $\ell\ge1$, and let $\pi\in S_{2\ell}$. Then,
  \begin{enumerate}
    \item[(i)] If $\pi$ has a fixed point, then $\val(\pi)=0$.
    \item[(ii)] If $\pi$ has no fixed point, then
      $\big|\val(\pi)\big|\le\big(\Tr\Sigma_{\mathcal L}\big)^{\ell}$.
  \end{enumerate}
\end{lemma}
\begin{proof}
  By \Cref{lem:cycletrace}, for every $i\in[r]^{\ell}$, we have the cycle factorization
  \begin{equation}\label{eq:cycfactor}
    \Tr\big[\mathbf M(i)\,P_\pi\big]
    =\prod_{c\,\in\,\mathrm{cycles}(\pi)}\Tr\big[M_{j_1}M_{j_2}\cdots M_{j_{\ell_c}}\big],
  \end{equation}
  where each cycle $c$ is written as $c=(j_1\, \dots \, j_{\ell_c})$, with $\ell_c$ is its length.
  By \cref{eq:slots}, each factor $M_j$ equals $L_{i_a}$ or $\adj{L_{i_a}}$ for some $a\in[\ell]$.

  \medskip\noindent
  (i) Suppose $\pi$ fixes some $j$. Then $(j)$ is a cycle of length one, so the corresponding factor in \cref{eq:cycfactor} is $\Tr[M_j]$ where $M_j$ is either $L_{i_a}$ and $\adj{L_{i_a}}$.
  Both $L_{i_a}$ and $\adj{L_{i_a}}$ are traceless (the latter since $\Tr[\adj{L_i}]=\overline{\Tr L_i}=0$).
  Hence the product in \cref{eq:cycfactor} vanishes for every $i\in[r]^{\ell}$, and therefore $\val(\pi)=0$.

  \medskip\noindent
  (ii) Suppose $\pi$ has no fixed point, so that every cycle satisfies $\ell_c\ge2$. Fix
  $i\in[r]^{\ell}$ and a cycle $c$.
  \begin{align*}
    \Big|\Tr\big[M_{j_1}\cdots M_{j_{\ell_c}}\big]\Big|
    &\ \le\ \big\|M_{j_1}\cdots M_{j_{\ell_c}}\big\|_{1}
    &&\mbox{\Cref{fact:schatten}(i)}\\
    &\ \le\ \prod_{j\in c}\|M_j\|_{\ell_c}
    &&\mbox{\Cref{fact:schatten}(ii)}\\
    &\ \le\ \prod_{j\in c}\|M_j\|_{2}
    &&\mbox{\Cref{fact:schatten}(iii).}
  \end{align*}
  The right-hand side is symmetric in the factors. Multiplying over all cycles of $\pi$, using the fact that the cycles of $\pi$ form a partition of $\{1,\dots,2\ell\}$, and using \cref{eq:slots} together with \Cref{fact:schatten}(iv),
  \begin{equation} \label{eq:per_tuple_bound}
    \Big|\Tr\big[\mathbf M(i)\,P_\pi\big]\Big|
    \ \le\ \prod_{j=1}^{2\ell}\|M_j\|_{2}
    \ =\ \prod_{a=1}^{\ell}\|L_{i_a}\|_{2}\,\big\|\adj{L_{i_a}}\big\|_{2}
    \ =\ \prod_{a=1}^{\ell}\|L_{i_a}\|_{2}^{2}.
  \end{equation}

  Now,
  \begin{align*}
    |\val(\pi)|
    &=\Big|\sum_{i\in[r]^{\ell}}\Tr\big[\mathbf M(i)\,P_\pi\big]\Big|\\
    &\le\sum_{i\in[r]^{\ell}}\Big|\Tr\big[\mathbf M(i)\,P_\pi\big]\Big|
    &&\mbox{by the triangle inequality}\\
    &\le\sum_{i\in[r]^{\ell}}\ \prod_{a=1}^{\ell}\|L_{i_a}\|_{2}^{2}
    &&\mbox{by \cref{eq:per_tuple_bound}}\\
    &=\prod_{a=1}^{\ell} \sum_{i_a=1}^{r}\|L_{i_a}\|_{2}^{2}
    \\
    &=\Big(\sum_{i=1}^{r}\|L_{i}\|_{2}^{2}\Big)^{\ell}\\
    &=\Big(\sum_{i=1}^{r}\Tr[\adj{L_i}L_i]\Big)^{\ell}\\
    &=\big(\Tr\Sigma_{\mathcal L}\big)^{\ell}.
  \end{align*}
\end{proof}

\begin{remark}\label{rem:derange-essential}
  \Cref{lem:C}(ii) uses $\ell_c\ge2$ exactly once, to pass from $\|\cdot\|{\ell_c}$ to $\|\cdot\|_{2}$.
  A similar bound for permutations with fixed points would leave $|\Tr B|\le\|B\|_{1}$, and $\|B\|_{1}\le\sqrt d\,\|B\|_{2}$ reintroduces the dimension.
  But such permutations with fixed points vanish identically by Part~(i). The cycles the Hölder inequalities cannot handle are precisely the ones tracelessness kills.
\end{remark}

The bound in \Cref{lem:C}(ii) is tight. Take the permutation
\[
  \pi=\prod_{a=1}^{\ell}(2a-1\ \ 2a),
\]
which has no fixed point and whose cycles are exactly the pairs $\{2a-1,2a\}$. By \cref{eq:slots} the cycle $(2a-1\,2a)$ contributes $\Tr\big[L_{i_a}\adj{L_{i_a}}\big]=\|L_{i_a}\|_2^{2}$, so $\Tr[\mathbf M(i)P_{\pi}]=\prod_{a=1}^{\ell}\|L_{i_a}\|_2^{2}$ for every $i\in[r]^{\ell}$; summing over $i$ as above gives $\val(\pi)=(\Tr\Sigma_{\mathcal L})^{\ell}$ exactly. In particular, $\Tr\Sigma_{\mathcal L}$ is an optimal bound.

The preceding estimates can now be combined into the following general bound for traceless families.

\begin{theorem}[Master bound for traceless families]\label{thm:master}
  Let $\mathcal L=\{L_i\}_{i=1}^{r}$ be a finite family of traceless operators on $\C^d$ and let $\ell\ge1$.
  Then,
  \[
    \E_{v \leftarrow \mathbb{S}(\C^{d})}\big[Y_{\mathcal L}^{\ell}(v)\big]
    \;\le\;\frac{\big(\Tr\Sigma_{\mathcal L}\big)^{\ell}}{\binom{d+2\ell-1}{2\ell}}
    \;\le\;\big(\Tr\Sigma_{\mathcal L}\big)^{\ell}\Big(\frac{2\ell}{d}\Big)^{2\ell}.
  \]
\end{theorem}
\begin{proof}
  By \Cref{lem:C}, $|\val(\pi)|\le(\Tr\Sigma_{\mathcal L})^{\ell}$ for every $\pi\in S_{2\ell}$. Hence,
  \begin{align*}
    \E_{v \leftarrow \mathbb{S}(\C^{d})}\big[Y_{\mathcal L}^{\ell}(v)\big]
    &=\frac{1}{(2\ell)!\,\binom{d+2\ell-1}{2\ell}}\sum_{\pi\in S_{2\ell}}\val(\pi)
    &&\mbox{by \Cref{prop:permexp}}\\
    &\le\frac{1}{(2\ell)!\,\binom{d+2\ell-1}{2\ell}}\sum_{\pi\in S_{2\ell}}\big|\val(\pi)\big|
    &&\mbox{by the triangle inequality}\\
    &\le\frac{1}{(2\ell)!\,\binom{d+2\ell-1}{2\ell}}\sum_{\pi\in S_{2\ell}}\big(\Tr\Sigma_{\mathcal L}\big)^{\ell}
    &&\mbox{by \Cref{lem:C}}\\
    &=\frac{\big(\Tr\Sigma_{\mathcal L}\big)^{\ell}}{\binom{d+2\ell-1}{2\ell}}\\
    &\le\big(\Tr\Sigma_{\mathcal L}\big)^{\ell}\Big(\frac{2\ell}{d}\Big)^{2\ell}
    &&\mbox{by \Cref{fact:binom}.}
  \end{align*}
\end{proof}

\subsubsection{The traceless family associated with a phase-1 isometry}
\label{sec:phase1-traceless-family}

We next associate a traceless family with every bounded phase-1 isometry.  This family will be used to control the accepting branch of the decoder.

Fix a phase-1 isometry
\[
  V:\C^d_C\otimes\ket0_{E_0}\to\C^d_{C}\otimes\Hs_E
\]
of circuit size at most $p=p(\lambda,k)$, that is, an element of $\cG_p$.  Since $V$ is realized by a circuit of size $p$ over constant-size gates, $\Hs_E$ is finite-dimensional, write $\{\ket e\}_e$ for a fixed orthonormal basis of $\Hs_E$.

\medskip\noindent
\textbf{Isometry slices.}
Define the slices of $V$ by
\begin{equation}\label{eq:slices}
  A_e=\big(\id_{C}\otimes\bra e_E\big)V:\C^d_C\to\C^d_{C} .
\end{equation}
Because $V$ is an isometry,
\begin{equation}\label{eq:slice-tp}
  \sum_e \adj{A_e}A_e
  =\adj V\Big(\sum_e\id_{C}\otimes\proj e\Big)V
  =\adj VV
  =\id_C .
\end{equation}
Thus the slices satisfy the same completeness relation as Kraus operators.

\medskip\noindent
\textbf{The accept-branch vector.}
For $s,t\in\M$ define
\[
  w^{(t)}_{s}
  =\big(\bra{\psi_t}_{C}\otimes\id_E\big)V\ket{\psi_s}_C
  \in\Hs_E,
  \qquad
  w_s=w^{(s)}_s .
\]
Since $\Dec$ measures $C$ with the POVM $\{\proj{\psi_t}\}_t\cup\{\Pi_\perp\}$, the probability that the decoder outputs $t$ when $s$ was encoded is $\|w^{(t)}_s\|^2$, and the corresponding subnormalized state on $E$ is $w^{(t)}_s\adj{(w^{(t)}_s)}$.  In the notation of \cref{eq:branch-decomp},
\begin{equation}\label{eq:branch-states}
  \sigma_{\same}(s)=w_s\adj{w_s},
  \qquad
  \sigma_t(s)=w^{(t)}_s\adj{(w^{(t)}_s)}
  \quad(t\ne s).
\end{equation}

\medskip\noindent
\textbf{The reference vector.}
Define
\begin{equation}\label{eq:ref-vector}
  v_r=\frac1d\sum_i\big(\bra i\otimes\id_E\big)V\ket i,
  \qquad\text{so that}\qquad
  \langle e\,|\,v_r\rangle=\frac{\Tr A_e}{d} .
\end{equation}
The vector $v_r$ depends only on $V$ and involves neither a message nor a codeword.

\medskip\noindent
\textbf{Traceless parts of the slices.}
Split each slice into its homothety and traceless parts,
\begin{equation}\label{eq:slice-traceless}
  A_e=\frac{\Tr A_e}{d}\, I_C +L_e,
  \qquad
  L_e=A_e-\frac{\Tr A_e}{d}\, I_C ,
  \qquad
  \cL=\{L_e\}_e .
\end{equation}
Then $\cL$ is a finite traceless family on $\C^d_C$.

\begin{lemma}[Identification]\label{lem:identify}
  Let $V$ be a phase-1 isometry as above, with slices $\{A_e\}_e$ as in \cref{eq:slices} and traceless family $\cL=\{L_e\}_e$ as in \cref{eq:slice-traceless}. Then:
  \begin{enumerate}
    \item[(i)] For every $s\in\M$ and every $e$,
      \[
        \langle e\,|\,w_s\rangle
        =\bra{\psi_s}A_e\ket{\psi_s},
        \qquad
        \langle e\,|\,v_r\rangle
        =\frac{\Tr A_e}{d},
        \qquad
        \langle e\,|\,(w_s-v_r)\rangle
        =\bra{\psi_s}L_e\ket{\psi_s}.
      \]

    \item[(ii)] For every $s\in\M$,
      \[
        \|w_s-v_r\|^2
        =\sum_e\big|\bra{\psi_s}L_e\ket{\psi_s}\big|^2
        =Y_{\cL}(\psi_s).
      \]

    \item[(iii)] The total weight of $\cL$ is
      \[
        \Tr\Sigma_{\cL}
        =d\big(1-\|v_r\|^2\big)
        \le d.
      \]
  \end{enumerate}
\end{lemma}

\begin{proof}
  For (i), by \cref{eq:slices},
  \[
    \langle e\,|\,w_s\rangle
    =\big(\bra{\psi_s}\otimes\bra e\big)V\ket{\psi_s}
    =\bra{\psi_s}A_e\ket{\psi_s},
  \]
  and similarly
  \[
    \langle e\,|\,v_r\rangle
    =\frac1d\sum_i\bra iA_e\ket i
    =\frac{\Tr A_e}{d}.
  \]
  Subtracting componentwise and using \cref{eq:slice-traceless} gives the last identity.

  For (ii), expand the norm in the basis $\{\ket e\}_e$ and insert part~(i):
  \[
    \|w_s-v_r\|^2
    =\sum_e\big|\langle e\,|\,(w_s-v_r)\rangle\big|^2
    =\sum_e\big|\bra{\psi_s}L_e\ket{\psi_s}\big|^2
    =Y_{\cL}(\psi_s).
  \]

  For (iii), using \cref{eq:slice-tp},
  \begin{align*}
    \Tr\Sigma_{\cL}
    &=\sum_e\|L_e\|_2^2\\
    &=\sum_e\Tr[\adj{A_e}A_e]
    -\frac1d\sum_e|\Tr A_e|^2\\
    &=\Tr\Big[\sum_e\adj{A_e}A_e\Big]
    -\frac1d\sum_e|\Tr A_e|^2\\
    &=d-\frac1d\sum_e|\Tr A_e|^2\\
    &=d\big(1-\|v_r\|^2\big)\le d.
  \end{align*}
\end{proof}

\subsubsection{Approximate design estimates}

The results established in the preceding sub-sections are based on Haar-randomness. We next show that the Haar-random unitaries can be replaced by approximate unitary designs.

\begin{lemma}[The moment operator factorizes]\label{lem:tensorpower}
  Let $\cL=\{L_e\}_{e=1}^{r}$ be a finite family of traceless operators on $\C^d$ and let $\ell\ge1$.
  Define
  \[
    T \;=\; \sum_e L_e\otimes \adj{L_e}\ \in\ \cL\big(\C^d\otimes\C^d\big),
    \qquad
    N \;=\; \sum_{\mathbf i\in[r]^\ell} M(\mathbf i).
  \]
  Then, with the $2\ell$ slots grouped into the pairs $\{2a-1,2a\}$, we have
  \[
    N \;=\; T^{\otimes\ell},
    \qquad
    Y_{\cL}(v)^{\ell}=\Tr\big[N\,\proj{v}^{\otimes2\ell}\big],
    \qquad
    \text{ and }
    \qquad \|N\|_\infty=\|T\|_\infty^{\ell}\le\big(\Tr\Sigma_{\cL}\big)^{\ell}\le d^{\ell}.
  \]
\end{lemma}

\begin{proof}
  \emph{The factorization.}
  By \cref{eq:slots}, $M(\mathbf i)=\bigotimes_{a=1}^{\ell}\big(L_{i_a}\otimes\adj{L_{i_a}}\big)$: the pair of slots $\{2a-1,2a\}$ carries the index $i_a$ and no other.
  The $\ell$ indices therefore range independently over the $\ell$ pairs, so
  \begin{align*}
    N &= \sum_{i_1,\dots,i_\ell\in[r]}\ \bigotimes_{a=1}^{\ell}\big(L_{i_a}\otimes\adj{L_{i_a}}\big)
    &&\text{definition of $N$, \cref{eq:slots}}\\
    &= \bigotimes_{a=1}^{\ell}\Big(\sum_{e}L_e\otimes\adj{L_e}\Big)
    &&\text{the $\ell$ indices are independent}\\
    &= T^{\otimes\ell}
    &&\text{definition of $T$.}
  \end{align*}
  The identity $Y_{\cL}(v)^{\ell}=\Tr\big[N\proj{v}^{\otimes2\ell}\big]$ is \cref{eq:Ypower}.

  \medskip\noindent
  \emph{The norm of $T$.}
  \begin{align*}
    \|T\|_\infty
    &\le \sum_e\big\|L_e\otimes\adj{L_e}\big\|_\infty
    &&\text{triangle inequality}\\
    &= \sum_e\|L_e\|_\infty\,\big\|\adj{L_e}\big\|_\infty
    &&\text{$\|\cdot\|_\infty$ is multiplicative on tensors}\\
    &= \sum_e\|L_e\|_\infty^{2}
    &&\text{\Cref{fact:schatten}(iv)}\\
    &\le \sum_e\|L_e\|_2^{2}
    &&\text{\Cref{fact:schatten}(iii)}\\
    &= \Tr\Sigma_{\cL}
    &&\text{from definition of $\Sigma_{\cL}$}\\
    &\le d
    &&\text{\Cref{cor:fullTD,lem:identify}.}
  \end{align*}

  \medskip\noindent
  The operator norm is multiplicative under tensor products, so $\|N\|_\infty=\|T^{\otimes\ell}\|_\infty=\|T\|_\infty^{\ell}\le d^{\ell}$.
\end{proof}

We now extend the preceding Haar-moment estimate \Cref{thm:master} to approximate unitary designs.

\begin{theorem}[Master bound for approximate designs]\label{thm:masterapprox}
  Let $\cL$ be a finite family of traceless operators on $\C^d$, let $\ell\ge1$ with $2\ell<d$, and let $\nu$ be an $\varepsilon$-approximate $2\ell$-design on $\U(d)$ with $\varepsilon\;\le\;\Big(\frac{4\ell^{2}}{d^{2}}\Big)^{\ell}.$
  Then for every fixed unit vector $\ket s$,
  \[
    \E_{U\sim\nu}\Big[\,Y_{\cL}\big(U\ket s\big)^{\ell}\,\Big]\;\le\;2\Big(\frac{4\ell^{2}}{d}\Big)^{\ell}.
  \]
  In particular, $\log_2(1/\varepsilon)=\poly(n)$ suffices whenever $\ell=\poly(n)$.
\end{theorem}

\begin{proof}
  Write $\rho=\proj{Us}$ and
  \[
    \Delta\;=\;\E_{U\sim\nu}\big[\rho^{\otimes2\ell}\big]-\E_{U\sim\mathrm{Haar}}\big[\rho^{\otimes2\ell}\big].
  \]
  Since $\rho^{\otimes2\ell}=U^{\otimes2\ell}\proj{s}^{\otimes2\ell}\adj{(U^{\otimes2\ell})}$, the two expectations are the two twirls of \Cref{def:approxdesign} applied to the same fixed state, that is $\Delta=\big(\cM^{(2\ell)}_\nu-\cM^{(2\ell)}_{\mathrm{Haar}}\big)\big(\proj{s}^{\otimes2\ell}\big)$, and hence, $\|\Delta\|_1\le\varepsilon$.

  \medskip\noindent
  By \Cref{lem:tensorpower}, $Y_{\cL}(Us)^{\ell}=\Tr\big[N\rho^{\otimes2\ell}\big]$, so this quantity is linear in $\rho^{\otimes2\ell}$ and
  \begin{align*}
    \Big|\E_{\nu}\big[Y_{\cL}^{\ell}\big]-\E_{\mathrm{Haar}}\big[Y_{\cL}^{\ell}\big]\Big|
    &= \big|\Tr[N\Delta]\big|
    &&\text{linearity}\\
    &\le \|N\Delta\|_1
    &&\text{\Cref{fact:schatten}(i)}\\
    &\le \|N\|_\infty\,\|\Delta\|_1
    &&\text{\Cref{fact:schatten}(ii), $p=\infty$, $q=1$}\\
    &\le d^{\ell}\,\varepsilon
    &&\text{\Cref{lem:tensorpower} and $\|\Delta\|_1\le\varepsilon$.}
  \end{align*}
  Thus,
  \begin{align*}
    \E_{\nu}\big[Y_{\cL}^{\ell}\big]
    &\le \E_{\mathrm{Haar}}\big[Y_{\cL}^{\ell}\big]+d^{\ell}\varepsilon\\
    &\le \Big(\frac{4\ell^{2}}{d}\Big)^{\ell}+d^{\ell}\Big(\frac{4\ell^{2}}{d^{2}}\Big)^{\ell}
    &&\text{\Cref{thm:master}, hypothesis on $\varepsilon$}\\
    &= 2\Big(\frac{4\ell^{2}}{d}\Big)^{\ell}. &&
  \end{align*}
  Now, taking logarithms on $\varepsilon$, we get, $\log_2(1/\varepsilon)\ge\ell\big(2n-2\log_2\ell-2\big)$, which is $\poly(n)$ whenever $\ell=\poly(n)$.
\end{proof}

\begin{lemma}[Distinct-message moments under approximate designs]\label{lem:rankfreeapprox}
  Let $\Phi:\mathsf{L}(\C^d)\to\mathsf{L}(\C^d)$ be a quantum channel, let two messages $s,t\in\M$ such that $s\ne t$.
  Let $m\ge2$, and let $\nu$ be an $\varepsilon$-approximate $2m$-design on $\U(d)$.
  Then
  \[
    \E_{U\sim\nu}\big[\,X_{st}^{m}\,\big]\;\le\;\Big(\frac md\Big)^{m}+d^{m}\varepsilon .
  \]
\end{lemma}
\begin{proof}Fix a Kraus representation $\lbrace K_i\rbrace_i$ for $\Phi$.

  \medskip \noindent
  Write $\rho_s=\proj{\psi_s}$ and $\rho_t=\proj{\psi_t}$.
  Using \Cref{lem:cycletrace}, we have $\Tr[XY]=\Tr[(X\otimes Y)\SWAP]$, with $\SWAP$ the swap on $\C^d\otimes\C^d$, and then
  \[
    X_{st}=\Tr\big[\Phi(\rho_s)\,\rho_t\big]
    =\Tr\big[(\rho_s\otimes\rho_t)\,W\big],
    \qquad
    W=\sum_i\big(\adj{K_i}\otimes\id\big)\,\SWAP\,\big(K_i\otimes\id\big).
  \]
  The operator $W$ does not depend on the choice of Kraus representation: under $K_i\mapsto\sum_j u_{ij}K_j$ with $\adj u\,u=\id$ the resulting double sum collapses by orthogonality.
  Taking the $m$-th power, $(X_{st})^m=\Tr\big[W^{\otimes m}(\rho_s\otimes\rho_t)^{\otimes m}\big]$.

  \medskip\noindent
  \emph{Now, the norm of $W$.}
  \begin{align*}
    \|W\|_\infty
    &\le \sum_i\big\|(\adj{K_i}\otimes\id)\,\SWAP\,(K_i\otimes\id)\big\|_\infty
    &&\text{triangle inequality}\\
    &= \sum_i\big\|\adj{K_i}\otimes\id\big\|_\infty\big\|K_i\otimes\id\big\|_\infty
    &&\text{$\SWAP$ is unitary}\\
    &= \sum_i\|K_i\|_\infty^{2}
    &&\text{\Cref{fact:schatten}(iv)}\\
    &\le \sum_i\|K_i\|_2^{2}
    &&\text{\Cref{fact:schatten}(iii)}\\
    &= \Tr\Big[\sum_i\adj{K_i}K_i\Big]=d
    &&\text{trace preservation,}
  \end{align*}
  so $\|W^{\otimes m}\|_\infty=\|W\|_\infty^{m}\le d^{m}$.

  \medskip\noindent
  Since $(\rho_s\otimes\rho_t)^{\otimes m}=U^{\otimes2m}\big(\proj{s}\otimes\proj{t}\big)^{\otimes m}\adj{(U^{\otimes2m})}$, the difference of its $\nu$- and Haar-expectations is $\big(\cM^{(2m)}_\nu-\cM^{(2m)}_{\mathrm{Haar}}\big)$ applied to a fixed state, hence of trace norm at most $\varepsilon$.
  Arguing as in \Cref{thm:masterapprox},
  \[
    \E_\nu\big[(X_{st})^m\big]
    \;\le\;\E_{\mathrm{Haar}}\big[(X_{st})^m\big]+\|W^{\otimes m}\|_\infty\,\varepsilon
    \;\le\;\Big(\frac md\Big)^m+d^{m}\varepsilon,
  \]
  the last step is by \Cref{lem:rankfree}.
\end{proof}

\subsection{Completeness}

For the construction of \Cref{def:enc}, the decoder, which also knows $\pp=U$, applies $\adj U$. Hence, for every message $m$,
\[
  \Dec_U\!\left(\Enc_U(m)\right)
  =
  \adj U U \ket{0^n,m}
  =
  \ket{0^n,m},
\]
and therefore
\[
  \Pr\!\left[\Dec_U\!\left(\Enc_U(m)\right)=m\right]=1.
\]
Thus, the construction achieves perfect completeness. This establishes
the completeness branch \textup{(C)} of \Cref{def:pq}.

\subsection{Detection}

For the detection branch, we apply the off-diagonal moment estimates above to the finite family of channels $\Phi_V$ induced by the phase-1 circuits $V\in\cG_p$.

\begin{proposition}[Relaxed detection at moment order $m$]\label{prop:relaxedgeneral}
  Fix expansion factor $\gamma>2$, set $k=\frac{n}{\gamma}$ and $\varepsilon=2^{-k}$, and let $\FAdv$ be a finite family of channels on $\C^d$.
  Let $m$ be an integer with $2\le m<d$ satisfying
  \begin{equation}\label{eq:momentcondition}
    m\Big[\big(1-\tfrac{2}{\gamma}\big)n-\log_2 m\Big]
    \;\ge\; 2k+\log_2|\FAdv|+n .
  \end{equation}
  Then, for Haar-random $U$,
  \[
    \Pr_{U}\Big[\;\exists\, s\in\M,\ \exists\,\Phi\in\FAdv:\
    \textstyle\sum_{t\ne s}X_{st}>\varepsilon \Big]\;\le\;2^{-n}.
  \]
\end{proposition}
\begin{proof}
  Fix $s\in\M$ and $\Phi\in\FAdv$.
  Since there are $2^k-1$ non-negative terms $X_{st}$ with distinct messages $t\ne s$, set $\tau = 2^{-2k}$. If $X_{st}\le\tau$ for every $t\ne s$, then
  \[
    \sum_{t\ne s}X_{st}
    \le (2^k-1)\tau
    = (2^k-1)2^{-2k}
    < 2^{-k}
    = \varepsilon.
  \]
  In other words, if the total contribution of the off-diagonal terms exceeds $\varepsilon$, then at least one individual term must itself exceed the threshold $\tau$. Therefore,
  \begin{equation} \label{eq:event-inclusion}
    \Big\{\textstyle\sum_{t\ne s}X_{st} > \varepsilon\Big\}
    \subseteq
    \bigcup_{t\ne s}\{X_{st}>\tau\}.
  \end{equation}
  Here, the braces $\{\cdot\}$ denote events with respect to the randomness of the experiment: the event on the left is that $\sum_{t\ne s}X_{st}$ exceeds $\varepsilon$, while the event on the right is that $X_{st}$ exceeds $\tau$.

  By \Cref{lem:rankfree} and $m\ge2$ we have $\E_U[(X_{st})^m]\le (m/d)^m$, and $d\tau=2^{\,n-2k}=d^{\,1-2/\gamma}$, so Markov's inequality at order $m$ gives
  \[
    \Pr_U[X_{st}>\tau]\;\le\;\frac{\E_U[(X_{st})^m]}{\tau^m}
    \;\le\;\Big(\frac{m}{d\tau}\Big)^{m}
    \;=\;\Big(\frac{m}{d^{\,1-2/\gamma}}\Big)^{m}.
  \]

  Define the bad event:
  \[
    \mathrm{BAD}
    =
    \Big\{
      \exists\,s\in\M,\ \exists\,\Phi\in\FAdv:
      \textstyle\sum_{t\ne s}X_{st}>\varepsilon
    \Big\}.
  \]
  There are fewer than $2^{2k}$ ordered pairs $(s,t)$ with $t\ne s$,
  and $|\FAdv|$ channels. Thus, by the preceding inclusion \cref{eq:event-inclusion} and a union
  bound,
  \[
    \Pr_U[\mathrm{BAD}]
    \le
    2^{2k}|\FAdv|
    \Big(\frac{m}{d^{\,1-2/\gamma}}\Big)^m.
  \]
  Taking logarithms,
  \begin{align*}
    \log_2\Pr_U[\mathrm{BAD}]
    &\le
    2k+\log_2|\FAdv|
    +m\Big[\log_2m-\big(1-\tfrac{2}{\gamma}\big)n\Big]\\
    &\le
    2k+\log_2|\FAdv|
    -\big(2k+\log_2|\FAdv|+n\big)\\
    &=-n,
  \end{align*}
  where the second inequality is exactly \cref{eq:momentcondition}.
  Therefore,
  \[
    \Pr_U[\mathrm{BAD}]\le2^{-n},
  \]
  as claimed.
\end{proof}

We now apply \Cref{prop:relaxedgeneral} to the finite family
\[
  \FAdv=\{\Phi_V:V\in\cG_p\}.
\]
By the choice of moment order in \cref{def:momentorder} and
\Cref{fact:gatecount}, the condition \cref{eq:momentcondition} is
satisfied for all sufficiently large $n$. Hence, except with probability
at most $2^{-n}$ over the Haar-random choice of $U$, simultaneously for
every $s\in\M$ and every $V\in\cG_p$,
\[
  \sum_{t\ne s}X_{st}\le\varepsilon.
\]

Now fix a $p$-bounded adversary $\A$ and
$b\in\{0,1\}$. For every public parameter $\pp=U$, the message $t=m_b$ chosen
by the adversary may depend on $U$, as may its phase-1 circuit $V_\pp$,
but $V_\pp\in\cG_p$. Therefore, on the event above,
\[
  \Pr\!\left[
    \hat m\notin\{\perp,m_b\}
    \,\middle|\,
    U
  \right]
  =
  \sum_{s\ne t}X_{s t}
  \le\varepsilon.
\]
Averaging over $U$ and accounting for the bad event gives
\[
  \Pr\Big[
    \hat m\notin\{\perp,m_b\}
    \text{ in }
    \Expt^{\NTDshort}_{\Pi,\A}(\lambda,k,b)
  \Big]
  \le
  \varepsilon+2^{-n}
\]
Thus, for the chosen parameters, the right-hand side is negligible,
which establishes the detection condition \textup{(D)} of
\Cref{def:pq}.

\subsection{Everlasting secrecy}\label{sec:secrecy}

It remains to establish the everlasting-secrecy branch of \Cref{def:pq}.  We use the traceless family associated with a phase-1 isometry introduced above.  By \Cref{lem:identify}, for every bounded phase-1 isometry $V\in\cG_p$, the accepting-branch vector $w_s$ is close to the message-independent reference vector $v_r$ whenever the corresponding quadratic statistic $Y_{\cL}(\psi_s)$ is small.  The master moment bound \Cref{thm:master} shows that this happens simultaneously for every message and every bounded phase-1 circuit with overwhelming probability over the choice of the public parameter.

\begin{lemma}[Concentration of the accept-branch vector]\label{lem:conc}
  Let $V\in\cG_p$ be fixed and let moment order $\ell$ be as in \cref{def:momentorder}, with $2\ell<d$. For Haar-random $U$ and every
  $s\in\M$,
  \[
    \E_U\big[\|w_s-v_r\|^{2\ell}\big]
    \le \Big(\frac{4\ell^2}{d}\Big)^\ell.
  \]
  Consequently, setting $\eta=2\ell\sqrt{2/d}$, we have
  \[
    \Pr_U\big[\|w_s-v_r\|>\eta\big]\le2^{-\ell}.
  \]
\end{lemma}
\begin{proof}
  By \Cref{lem:identify}(ii),
  \[
    \|w_s-v_r\|^{2\ell}
    =Y_{\cL}(\psi_s)^\ell,
  \]
  where $\psi_s=U\ket s$.  By \Cref{lem:identify}(iii),
  $\Tr\Sigma_{\cL}\le d$, and hence \Cref{thm:master} gives
  \[
    \E_U\big[\|w_s-v_r\|^{2\ell}\big]
    \le d^\ell\Big(\frac{2\ell}{d}\Big)^{2\ell}
    =\Big(\frac{4\ell^2}{d}\Big)^\ell.
  \]
  Markov's inequality at order $2\ell$, with
  $\eta^2=8\ell^2/d$, gives
  \[
    \Pr_U[\|w_s-v_r\|>\eta]
    \le
    \frac{(4\ell^2/d)^\ell}{(8\ell^2/d)^\ell}
    =2^{-\ell}.
  \]
\end{proof}

The preceding estimate holds for a fixed message and a fixed phase-1 circuit.  Our choice of moment order allows us to make it simultaneous over all messages and all phase-1 circuits of size at most $p$.

\begin{lemma}[Uniform concentration over bounded phase-1 circuits]
  \label{lem:uniform}
  Let $p$ be a polynomial, let $\ell$ be the moment order as in \cref{def:momentorder}, and put $\eta=2\ell\sqrt{2/d}$. Assume $n$ is large enough that $2\ell<d$.  Then, except with probability at most $2^{-n}$ over Haar-random $U$, for every $s\in\M$ and every $V\in\cG_p$,
  \[
    \|w_s(V)-v_r(V)\|\le\eta.
  \]
  Moreover $\eta=\poly(n)\cdot2^{-n/2}$.
\end{lemma}
\begin{proof}
  Fix $s\in\M$ and $V\in\cG_p$.  By \Cref{lem:conc},
  \[
    \Pr_U\big[\|w_s(V)-v_r(V)\|>\eta\big]
    \le2^{-\ell}.
  \]
  There are $2^k$ messages and $|\cG_p|$ circuits, so a union bound gives
  \[
    \Pr_U\Big[
      \exists s\in\M,\exists V\in\cG_p:
      \|w_s(V)-v_r(V)\|>\eta
    \Big]
    \le
    2^{k+\log_2|\cG_p|-\ell}
    \le2^{-n},
  \]
  by the definition of $\ell$. Finally $\ell=\poly(n)$ by \Cref{fact:gatecount}, so $\eta=\poly(n)\cdot2^{-n/2}$.
\end{proof}

We can now complete the proof of everlasting secrecy. Recall that clause \textup{(ES)} of \Cref{def:pq} asks us to bound
\[
  \Big\|
  (\id_F\otimes\Lambda)\tilde\sigma_{\pp,V}(m_0)
  -
  (\id_F\otimes\Lambda)\tilde\sigma_{\pp,V}(m_1)
  \Big\|_1
\]
for every phase-2 channel $\Lambda$.  Since trace distance cannot increase under the channel $\id_F\otimes\Lambda$, it is enough to establish the desired bound before phase~2, namely for $\Lambda=\id$.

We condition on two good events for the Haar-random public parameter $U$.  First, by the detection result established in the preceding subsection, except with probability at most $2^{-n}$ we have
\begin{equation}\label{eq:tdgood}
  \sum_{t\ne s}X_{st}
  \le 2^{-k}
  \qquad
  \text{for every }s\in\M
  \text{ and every }V\in\cG_p .
\end{equation}
Second, by \Cref{lem:uniform}, except with probability at most $2^{-n}$ we have
\begin{equation}\label{eq:esgood}
  \|w_s(V)-v_r(V)\|
  \le \eta
  \qquad
  \text{for every }s\in\M
  \text{ and every }V\in\cG_p ,
\end{equation}
where $\eta=\poly(n)\cdot 2^{-n/2}$.
By a union bound, both \cref{eq:tdgood,eq:esgood} therefore hold simultaneously except with probability at most $2^{-n+1}$.

The important point is that both good events hold uniformly over all $V\in\cG_p$.  Thus they remain valid even when the phase-1 circuit is chosen as a function of the public parameter $U$, as allowed in the security experiment.  We henceforth fix a public parameter $U$ for which both good events hold, fix the corresponding phase-1 isometry $V\in\cG_p$, and fix arbitrary messages $m_0,m_1\in\M$.  All quantities below are formed with respect to this fixed $U$ and $V$.

Because the flag register $F$ is classical and the decoder outcomes are mutually orthogonal, \cref{eq:nonabort} gives
\begin{multline}\label{eq:es-split}
  \|\tilde\sigma_{\pp,V}(m_0)
  -\tilde\sigma_{\pp,V}(m_1)\|_1
  \le
  \underbrace{
    \|\sigma_{\same}(m_0)-\sigma_{\same}(m_1)\|_1
  }_{(\mathrm I)} \\
  +
  \underbrace{
    \sum_{t\notin\{\perp,\same\}}
    \|\sigma_t(m_0)-\sigma_t(m_1)\|_1
  }_{(\mathrm{II})}.
\end{multline}

\medskip\noindent
\emph{The term $(\mathrm{II})$: detection.}
Each $\sigma_t(m_b)$ is positive semidefinite, so
\begin{align*}
  (\mathrm{II})
  &\le\sum_{t\notin\{\perp,\same\}}\Big(\Tr\sigma_t(m_0)+\Tr\sigma_t(m_1)\Big)
  &&\text{triangle inequality, positivity}\\
  &=\Pr\big[f\notin\{\perp,\same\}\bigm| m_0\big]+\Pr\big[f\notin\{\perp,\same\}\bigm| m_1\big]
  &&\text{$\Tr\sigma_t(m)$ is the weight of branch $t$}\\
  &\le2\cdot2^{-k}
  &&\text{by \cref{eq:tdgood}.}
\end{align*}

\medskip\noindent
\emph{The term $(\mathrm{I})$: concentration.}
By \cref{eq:branch-states} the accept branch is the rank-one subnormalized state $\sigma_{\same}(m_b)=w_{m_b}\adj{w_{m_b}}$, and telescoping gives
\[
  w_{m_0}\adj{w_{m_0}}-w_{m_1}\adj{w_{m_1}}
  =w_{m_0}\adj{\big(w_{m_0}-w_{m_1}\big)}+\big(w_{m_0}-w_{m_1}\big)\adj{w_{m_1}} .
\]
Hence
\begin{align*}
  (\mathrm{I})
  &\le\big\|w_{m_0}\adj{(w_{m_0}-w_{m_1})}\big\|_1+\big\|(w_{m_0}-w_{m_1})\adj{w_{m_1}}\big\|_1
  &&\text{triangle inequality}\\
  &=\leq \big(\|w_{m_0}\|+\|w_{m_1}\|\big)\,\big\|w_{m_0}-w_{m_1}\big\|
  &&\|a\adj b\|_1=\|a\|\,\|b\|\\
  &\le2\,\big\|w_{m_0}-w_{m_1}\big\|
  &&\|w_{m_b}\|^{2}=\Tr\sigma_{\same}(m_b)\le1\\
  &\le2\Big(\big\|w_{m_0}-v_r\big\|+\big\|v_r-w_{m_1}\big\|\Big)
  \\
  &\le4\eta
  &&\text{by \cref{eq:esgood}.}
\end{align*}

\medskip\noindent
Putting it back in \cref{eq:es-split}, on the good event and for every $V\in\cG_p$ and all $m_0,m_1\in\M$,
\[
  \big\|\tilde\sigma_{\pp,V}(m_0)-\tilde\sigma_{\pp,V}(m_1)\big\|_1
  \;\le\;4\eta+2\cdot2^{-k}
  \;=\;\poly(n)\cdot2^{-n/2}+2^{-n/\gamma+1}
  \;=\;2^{-\Omega(n)},
\]
the second term dominating because $\gamma>2$.

\medskip\noindent
Clause \textup{(ES)} requires the bound in expectation over $\pp=U$.
The trace distance between two subnormalized states is at most $2$, so splitting on the good event (where the bad event has probability at most $2^{-n+2}$),
\[
  \E_{U}\big\|\tilde\sigma_{\pp,V}(m_0)-\tilde\sigma_{\pp,V}(m_1)\big\|_1
  \;\le\;2^{-\Omega(n)}+2\cdot2^{-n+2}
  \;=\;2^{-\Omega(n)}.
\]

\subsection{Efficiency} \label{sec:efficient-pq}

It remains to replace the Haar-random unitary by an efficiently samplable ensemble. The detection proof uses the $2m$-th unitary moment and, for all sufficiently large $n$, we may take $m=\ell$ in \cref{eq:momentcondition}. The everlasting-secrecy proof uses the $2\ell$-th moment. Thus no moment above order $2\ell$ is required.

Let $\nu$ be a $\varepsilon$-approximate unitary $2\ell$-design, where
\[
  \varepsilon
  =
  \left(\frac{\ell}{d^2}\right)^\ell.
\]
For $s=t$, \Cref{thm:masterapprox} gives
\[
  \E_{U\sim\nu}\Big[Y_{\cL}(U\ket s)^\ell\Big]
  \le
  2\left(\frac{4\ell^2}{d}\right)^\ell.
\]
For $s\ne t$, \Cref{lem:rankfreeapprox} gives
\[
  \E_{U\sim\nu}\big[X_{st}^{\ell}\big]
  \le
  \left(\frac{\ell}{d}\right)^\ell
  +d^\ell\varepsilon
  =
  2\left(\frac{\ell}{d}\right)^\ell.
\]
Thus, in both cases, the corresponding moment bounds increase by at most a factor of $2$. Since this constant factor does not affect negligibility, the proofs of \textup{(D)} and \textup{(ES)} remain valid. Perfect completeness continues to hold for every unitary in the support of $\nu$.

It remains to bound the complexity of sampling and implementing $U$. By \Cref{fact:gatecount} and \cref{def:momentorder},
\[
  \ell
  =
  O\!\bigl(n+p\log(n+p)\bigr).
\]
Moreover, since $d=2^n$,
\[
  \log_2(1/\varepsilon)
  =
  \ell\bigl(2n-\log_2\ell\bigr)
  \le
  2n\ell.
\]
In particular, $\ell=\poly(n,p)$ and $\log_2(1/\varepsilon)=\poly(n,p)$. Under our choice of parameters, $p=\poly(n)$, and hence $\ell=\poly(n)$. It follows that $2\ell\le2^{n/4}$ for all sufficiently large $n$.

Efficient constructions of approximate unitary designs yield a $\varepsilon$-approximate unitary $t$-design, for $t\le2^{n/4}$, by quantum circuits of size \cite{MPSY24}
\[
  O\!\left(t\,\poly(n)+t\log(1/\varepsilon)\right).
\]
Taking $t=2\ell$ therefore gives circuits of size
\begin{equation}\label{eq:designsize}
  O\!\left(\ell\,\poly(n)+n\ell^2\right)
  =
  \poly(n,p).
\end{equation}

It remains to fix the polynomial $q$. From \cref{eq:designsize}, let
$C>0$ be a sufficiently large constant such that $\Setup$, $\Enc$, and
$\Dec$ are implementable by circuits of size at most
\[
  C\left(\ell\poly(n)+n\ell^2\right).
\]
We therefore set
\[
  q(\lambda,k)
  =
  C\left(\ell\poly(n)+n\ell^2\right).
\]
By \Cref{fact:gatecount} and \cref{def:momentorder},
\[
  \ell
  =
  O\!\left(n+p\log(n+p)\right).
\]
Substituting this into the definition of $q$ gives
\[
  q(\lambda,k)
  =
  O\!\left(
    \bigl(n+p\log(n+p)\bigr)\poly(n)
    +
    n\bigl(n+p\log(n+p)\bigr)^2
  \right).
\]
In particular, since $p=\poly(n)$, we have
\[
  q=\poly(n,p),
\]
and hence $q$ is polynomial in $(\lambda,k)$. This completes the proof
of \Cref{thm:pq-main}.
    \section{Relaxed tamper detection}

We now return to the relaxed tamper-detection notion. The proof uses the rank-free distinct-message estimate already established as a preliminary result for the proof of \Cref{thm:pq-main}

\begin{definition}[Tamper detection] \label{def:td}
  A code $(\Enc_U,\Dec_U)$ is \emph{$\varepsilon$-tamper-detecting} against $\Fadv$ if for every $m\in\mathcal M$ and every $\Phi\in\Fadv$
  \[
    \Pr_U[\Dec_U(\Phi(\Enc_U(m)))\in\{\perp\}]\ge1-\varepsilon,
  \]
  whenever $\Phi$ acts nontrivially, and it is \emph{$\varepsilon$-relaxed-tamper-detecting} if
  \[
    \Pr_U[\Dec_U(\Phi(\Enc_U(m)))\in\{\perp,m\}]\ge1-\varepsilon.
  \]
  We say the code is secure \emph{with overwhelming probability} if the above holds for a $1-\negl(k)$ fraction of $U$.
\end{definition}

In the relaxed setting, prior work conjectured that a cardinality bound alone should suffice, with no assumption on Kraus rank or entanglement fidelity, for every $\alpha<1$.  This is the universal relaxed-detection problem. Prior to the present work, the conjecture was known only for certain structured families of channels, including replacement and measure-and-prepare channels.

\begin{conjecture}[{\cite[Conj.~1]{BKR26}}]\label{conj:BKR26}
  Let $\Fadv$ be any family of channels with $|\Fadv|\le2^{d^\alpha}$ for some constant $\alpha<1$. Then there is an $\varepsilon$-relaxed-tamper-detection code against $\Fadv$ with $\varepsilon=\negl(k)$ and expansion factor $\gamma=\mathcal{O}(1)$.
\end{conjecture}

We are now able to prove this conjecture, establishing universal relaxed tamper detection against arbitrary families of channels with cardinality $2^{d^\alpha}$ for any $\alpha<1$.

\begin{theorem}[Universal Relaxed Quantum Tamper Detection]\label{thm:utd}
  Fix parameters $\alpha < 1$ and expansion factor $\gamma>\tfrac{2}{1-\alpha}$.
  And fix any  adversarial family $\Fadv$ of channels with $|\Fadv|\le2^{d^\alpha}$.
  Let $\left(\Enc_U, \dec_U\right)$ be the Haar random encoding-decoding strategy on $n$ qubits (as given in \Cref{def:enc}).
  Then, there exists an integer $n_0=n_0(\alpha,\gamma)$ such that for all $n\ge n_0$,
  \begin{equation*}
    \Pr_U \big[ (\enc_U, \dec_U) \text{ is $\varepsilon$-relaxed tamper secure against $\advCPTP$} \big] \geq 1 - \operatorname{negl}(k),
  \end{equation*}
  with $\varepsilon = 2^{-k}$.
  In particular, there exists  a unitary $U \in \mathcal{U}_d(\CC)$ such that $(\enc_U, \dec_U)$ is $\varepsilon$-relaxed tamper  secure against $\advCPTP$.
\end{theorem}
\begin{proof}
  We follow the same reduction as in the proof of \Cref{prop:relaxedgeneral}. Set $\varepsilon=2^{-k}$ and $\tau=2^{-2k}$. For every fixed $s\in\M$ and $\Phi\in\FAdv$, \cref{eq:event-inclusion} gives
  \[
    \Big\{\textstyle\sum_{t\ne s}X_{st}>\varepsilon\Big\}
    \subseteq
    \bigcup_{t\ne s}\{X_{st}>\tau\}.
  \]
  Moreover, for every $2\le m<d$, \Cref{lem:rankfree} and Markov's inequality give
  \[
    \Pr_U[X_{st}>\tau]
    \le
    \Big(\frac{m}{d^{\,1-2/\gamma}}\Big)^m.
  \]

  Define the bad event
  \[
    \mathrm{BAD}
    =
    \Big\{
      \exists\,s\in\M,\ \exists\,\Phi\in\FAdv:
      \textstyle\sum_{t\ne s}X_{st}\varepsilon
    \Big\}.
  \]
  There are fewer than $2^{2k}$ ordered pairs $(s,t)$ with $t\ne s$, and $|\FAdv|\le 2^{d^\alpha}$ channels. Thus, by \cref{eq:event-inclusion} and a union bound,
  \[
    \Pr_U[\mathrm{BAD}]
    \le
    2^{2k}2^{d^\alpha}
    \Big(\frac{m}{d^{\,1-2/\gamma}}\Big)^m.
  \]

  We now choose
  \[
    m=\lceil d^\alpha\rceil.
  \]
  For all sufficiently large $n$, we have $2\le m<d$. Taking logarithms in the preceding bound gives
  \begin{align*}
    \log_2\Pr_U[\mathrm{BAD}]
    &\le
    2k+d^\alpha
    +m\Big[\log_2m-\big(1-\tfrac{2}{\gamma}\big)n\Big].
  \end{align*}
  Since $m=\lceil d^\alpha\rceil$, we have
  \[
    \log_2m\le \alpha n+1.
  \]
  Let $\delta_\star = 1-\alpha-\frac{2}{\gamma}$. By the assumption $\gamma>\frac{2}{1-\alpha}$, we have
  $\delta_\star>0$, and hence
  \[
    \log_2m-\big(1-\tfrac{2}{\gamma}\big)n
    \le
    -\delta_\star n+1.
  \]
  In particular, for all sufficiently large $n$, the right-hand side is non-positive. Since $m\ge d^\alpha$, it follows that
  \begin{align*}
    \log_2\Pr_U[\mathrm{BAD}]
    &\le
    \frac{2n}{\gamma}
    +d^\alpha
    +d^\alpha(-\delta_\star n+1)\\
    &=
    -\delta_\star n d^\alpha
    +2d^\alpha
    +\frac{2n}{\gamma}.
  \end{align*}

  It remains to show that the two positive terms can be absorbed into the leading negative term. For all sufficiently large $n$,
  \[
    2d^\alpha
    \le
    \frac{\delta_\star}{4}n d^\alpha
    \qquad\text{and}\qquad
    \frac{2n}{\gamma}
    \le
    \frac{\delta_\star}{4}n d^\alpha.
  \]
  Therefore,
  \[
    \log_2\Pr_U[\mathrm{BAD}]
    \le
    -\frac{\delta_\star}{2}n d^\alpha,
  \]
  and consequently
  \[
    \Pr_U[\mathrm{BAD}]
    \le
    2^{-\frac{\delta_\star}{2}n d^\alpha}
    =
    \negl(n).
  \]

  Thus, except with negligible probability over the Haar-random choice of $U$, simultaneously for every $s\in\M$ and every $\Phi\in\FAdv$,
  \[
    \sum_{t\ne s}X_{st}
    \le
    \varepsilon.
  \]
  By \Cref{lem:dropdiag}, this means that $(\Enc_U,\Dec_U)$ is $\varepsilon$-relaxed-tamper-detecting against $\FAdv$. Since
  \[
    \varepsilon=2^{-k}=\negl(k),
  \]
  all but a negligible fraction of Haar-random unitaries yield a relaxed tamper-detection code against $\FAdv$. In particular, at least one such unitary exists.

  Finally, for every constant $\alpha<1$, one may choose any constant $\gamma>\frac{2}{1-\alpha}$, and hence the expansion factor is $\gamma=\mathcal{O}(1)$.
\end{proof}
    \section{Tamper detection}

We next strengthen relaxed detection to full tamper detection. The off-diagonal branch is already controlled by \Cref{thm:utd}, the remaining work is the diagonal overlap, for which the master bound \Cref{thm:master} from the proof of \Cref{thm:pq-main} is now reused.

The diagonal estimate depends not only on the cardinality of the adversarial family, but also on structural properties of the individual channels.

\begin{definition}[Bounded family] \label{def:bounded}
  A family $\Fadv$ of channels on $\C^d$ is \emph{$(\alpha,\delta,\delta')$-bounded} if the size of the family is at most $|\Fadv|\le2^{d^{\alpha}}$, if every $\Phi\in\Fadv$ has Kraus rank at most $\rank(\Phi)\le d^{1-\delta}$, and if every $\Phi\in\Fadv$ has entanglement fidelity at most $\Fe(\Phi)\le d^{-2\delta'}$.
\end{definition}

We record the notion of full tamper-detection that will be used throughout this section.

\begin{definition}[Tamper detection]\label{def:tamper-detection}
    Let $\FAdv$ be a family of channels on $\C^d$ and let
    $\varepsilon>0$. A code $(\Enc,\Dec)$ is
    $\varepsilon$-\emph{tamper-detecting} against $\FAdv$ if, for
    every message $m\in\mathcal M$ and every $\Phi\in\FAdv$,
    \[
        \Pr\!\left[
            \Dec\!\left(\Phi(\Enc(m))\right)=\perp
        \right]
        \ge 1-\varepsilon.
    \]
\end{definition}

We also recall the prior best known encoding and decoding for tamper detection. Detection against arbitrary CPTP families was first established under three simultaneous restrictions on the family cardinality, Kraus rank, and entanglement fidelity (\Cref{def:bounded}).

\begin{theorem}[{\cite[Thm.~4]{BKR26}}]\label{thm:BKR264}
  If $\Fadv$ is $(\alpha,\delta,\delta')$-bounded with $0\le2\alpha<\delta\le1$ and $\delta'>\tfrac{1-\delta}{2}$, then for expansion factor $\gamma>\max\{\tfrac{2}{\delta+2\delta'-1},\tfrac{2}{1-2\delta'-2\alpha}\}$ the Haar encoding is $\negl(k)$-tamper-detecting against $\Fadv$ with overwhelming probability.
\end{theorem}

Our \Cref{thm:fullTDfinal} below removes the Kraus-rank restriction entirely: we show that full tamper detection can be obtained under assumptions only on the cardinality of the adversarial family and on the entanglement fidelity of its channels.

The Following first order moment estimate Lemma is in \cite{BKR26}, we include it below for completeness, and refer the reader to \cite{BKR26} for the proof.

\begin{lemma}[Diagonal first moment {\cite[Lem.~5]{BKR26}}]\label{lem:diag1}
  For every message $s$ and every quantum channel $\Phi:\mathsf{L}(\C^d)\to\mathsf{L}(\C^d)$, the diagonal overlap $X_{ss}$ satisfies
  \[
    \E_U[X_{ss}]
    =
    \Favg(\Phi)
    =
    \frac{d\,\Fe(\Phi)+1}{d+1}.
  \]
\end{lemma}

We now refine the first-moment estimate of \Cref{lem:diag1} into a higher-moment bound that decays with the moment order and can therefore be combined with a union bound over the adversarial family.

Note that since $X_{ss}\in[0,1]$, this already gives a bound $\E[(X_{ss})^\ell]\le\Favg(\Phi)$.
However, such an $\ell$ independent bound is too weak for a union bound.
We need decay in the moment $\ell$. Recall that $\lbrace K_i \rbrace_i$ are the Kraus operators for the map $\Phi$. Write the decomposition $K_i=c_i I +L_i$ where $I$ is the identity operator, with $c_i=\frac{\Tr K_i}{d}$ and $L_i=K_i-c_i I $ the traceless part of the Kraus operator. Then
\begin{align*}
  X_{ss} & = \sum_{i} \big\vert \bra{\psi_s} K_i \ket{\psi_s} \big\vert ^2 \\
  & =  \sum_{i} \big\vert \bra{\psi_s} (c_i  I  + L_i)  \ket{\psi_s} \big\vert^2 \\
  & =  \sum_i \bra{\psi_s} (c_i I  + L_i)  \ket{\psi_s} \bra{\psi_s} \  (\overline{c_i} I  + \adj{L_i})  \ket{\psi_s} \\
  & = \sum_{i} \Big[ c_i \braket{\psi_s}{\psi_s} + \bra{\psi_s} L_i \ket{\psi_s} \Big] \Big[ \overline{c_i} \braket{\psi_s}{\psi_s} + \bra{\psi_s} \adj{L_i} \ket{\psi_s} \Big] \\
  & = \sum_{i} \vert c_i \vert^2 + c_i \bra{\psi_s} \adj{L_i} \ket{\psi_s} + \overline{c_i} \bra{\psi_s} {L_i} \ket{\psi_s} + \big\vert  \bra{\psi_s} L_i \ket{\psi_s}\vert^2 \\
  & = \sum_{i} \vert c_i \vert^2 + 2 \cdot \mathrm{Re}\left( \overline{c_i} \bra{\psi_s} {L_i} \ket{\psi_s}\right)  + \big\vert \bra{\psi_s} L_i \ket{\psi_s} \big\vert^2
\end{align*}
For ease of notation, define:
\begin{equation} \label{eq:def_W_and_Y}
  W_s=2\sum_i\mathrm{Re}\big(\overline{c_i}\bra{\psi_s} L_i \ket{\psi_s}\big) \quad \text{and}
  \qquad
  Y_s=\sum_i\big|\bra{\psi_s} L_i \ket{\psi_s}\big|^2.
\end{equation}

Thus, since the entanglement fidelity is $\sum_i|c_i|^2=\Fe(\Phi)$, we have the decomposition of $X_{ss}$ into the entanglement fidelity and two additional terms:
\begin{equation}\label{eq:Xss-split}
  X_{ss}=\Fe+W_s+Y_s.
\end{equation}
Two remarks are in order for \cref{eq:Xss-split}.
The second term $W_s$ has mean zero, $\E_U[W_s]=0$: for each traceless $L_i$,
\[
  \E_U\big[\bra{\psi_s}L_i\ket{\psi_s}\big]
  =\Tr\!\big[L_i\,\E_U[\proj{\psi_s}]\big]
  =\frac{\Tr L_i}{d}=0,
\]
using $\E_U[\proj{\psi_s}]= I /d$ (the $\ell=1$ case of \Cref{lem:sym-moment}).
The third term $Y_s$ collects the diagonal parts.
Because $\psi_s=U\ket s$ is Haar-distributed for every fixed $s$, the laws of $W_s$ and $Y_s$ under $U$ do not depend on $s$; we drop the subscript henceforth, writing $W$ and $Y$ for the same functionals of a generic Haar vector $v$.

\begin{lemma}[Rank-free diagonal bound]\label{cor:fullTD}
  Let $\Phi:\mathsf{L}(\C^d)\to\mathsf{L}(\C^d)$ be any quantum channel with Kraus operators $\{K_i\}_{i=1}^{r}$, let $c_i=\frac{\Tr K_i}{d}$ and $L_i=K_i-c_i I$ be their traceless parts, and let $Y_s$ be as in \cref{eq:def_W_and_Y}.
  Then for every $s\in\mathcal M$ and every $\ell\ge1$,
  \[
  \E_U\big[Y_s^{\ell}\big]\;\le\;\Big(\frac{4\ell^{2}}{d}\Big)^{\ell}.\]
\end{lemma}
\begin{proof}
  Recall the notations $Y_{\mathcal L}(v)=\sum_i\big|\bra v L_i\ket v\big|^2$ and $\Sigma_{\mathcal L}=\sum_i \adj{L_i}L_i$ from \Cref{def:traceless-family}.
  Each $L_i$ is traceless by construction, and $\psi_s=U\ket s$ is Haar-distributed on
  $\mathbb{S}(\C^{d})$, so $\E_U[Y_s^{\ell}]=\E_{v\leftarrow\mathbb{S}(\C^{d})}[Y^{\ell}_{\mathcal L}(v)]$
  and \Cref{thm:master} applies with $\mathcal L=\{L_i\}_{i=1}^{r}$. We first compute the total
  weight $\Tr\Sigma_{\mathcal L}$:
  \begin{align*}
    \Tr\big(\Sigma_{\mathcal L}\big)
    &=\sum_i\Tr\big[\adj{L_i}L_i\big]\\
    &=\sum_i\Big(\Tr[\adj{K_i}K_i]-c_i\overline{\Tr K_i}-\overline{c_i}\Tr K_i+d|c_i|^2\Big)
    &&\mbox{expanding $L_i=K_i-c_i I$}\\
    &=\sum_i\Tr[\adj{K_i}K_i]-d\sum_i|c_i|^{2}
    &&\mbox{$\Tr K_i=d\,c_i$}\\
    &=d-d\sum_i|c_i|^{2}
    &&\mbox{$\sum_i \adj{K_i}K_i= I$}\\
    &=d-\frac{1}{d}\sum_i\big|\Tr K_i\big|^{2}
    &&\mbox{$c_i=\Tr K_i/d$}\\
    &=d\,(1-\Fe)
    &&\mbox{by \cref{eq:Fe}}\\
    &\le d
    &&\mbox{$\Fe\ge0$.}
  \end{align*}
  Hence, by \Cref{thm:master},
  \[
    \E_U\big[Y_s^{\ell}\big]\; \le\big(\Tr\Sigma_{\mathcal L}\big)^{\ell}\Big(\frac{2\ell}{d}\Big)^{2\ell} \;\le\;d^{\ell}\Big(\frac{2\ell}{d}\Big)^{2\ell}
  =\Big(\frac{(2\ell)^{2}}{d}\Big)^{\ell}=\Big(\frac{4\ell^{2}}{d}\Big)^{\ell}.\]
\end{proof}

\begin{lemma}[Diagonal moment bound]\label{lem:diagmoment}
  Let $\Phi:\mathsf{L}(\C^d)\to\mathsf{L}(\C^d)$ be any quantum channel. For every $s\in\mathcal M$ and every $\ell\ge1$,
  \[
    \E_U\big[X_{ss}^{\ell}\big]\ \le\ 3^{\ell-1}\Bigg(\sqrt{\Fe(\Phi)}+\frac{2\ell}{\sqrt d}\Bigg)^{2\ell}.
  \]
\end{lemma}
\begin{proof}
  First let us prove the following inequality:
  for all integers $\ell \geq 1$ and for all $a,b,c \geq 0$, we have,
  \begin{equation}\label{fact:power_convex}
    (a+b+c)^\ell
    \leq 3^{\ell-1}(a^\ell+b^\ell+c^\ell)
  \end{equation}
  Indeed, since $x \mapsto x^\ell$ is convex on $[0,\infty)$ for every integer $\ell \geq 1$, Jensen's inequality gives
  \[
    \left(\frac{a+b+c}{3}\right)^\ell
    \leq \frac{a^\ell+b^\ell+c^\ell}{3}.
  \]
  Multiplying both sides by $3^\ell$ yields \cref{fact:power_convex}.

  Recall that by \cref{eq:Xss-split},  $X_{ss}=\Fe+W_s+Y_s$.
  We want to use \cref{fact:power_convex} on this, but $W_s$ can be negative, hence we write a weaker inequality $X_{ss} \leq \Fe+ \vert W_s \vert +Y_s$.
  Note that $\Fe$ and $Y_s$ are non-negative by definition.
  Using \cref{fact:power_convex} followed by taking expectation over $U$, we get,
  \begin{equation}\label{eq:threeterm}
    \E_U\big[(X_{ss})^{\ell}\big]\ \le\ 3^{\ell-1}\Big(\Fe^{\ell}+\E_U\big[|W_s|^{\ell}\big]+\E_U\big[Y_s^{\ell}\big]\Big).
  \end{equation}
  \Cref{cor:fullTD} already bounds the third term of \cref{eq:threeterm}.
  For the middle term, recall from \cref{eq:def_W_and_Y} that $W_s=2\sum_i\mathrm{Re}\big(\overline{c_i}\bra{\psi_s}L_i\ket{\psi_s}\big)$.
  So Cauchy--Schwarz gives,
  \[
    |W_s|
    \ \le\ 2\Big(\sum_i|c_i|^{2}\Big)^{1/2}\Big(\sum_i\big|\bra{\psi_s}L_i\ket{\psi_s}\big|^{2}\Big)^{1/2}
    \ =\ 2 \Fe^{\frac{1}{2}}\,Y_s^{\frac{1}{2}},
  \]
  using $\sum_i|c_i|^{2}=\Fe$ and the definition of $Y_s$.
  Hence,
  \begin{align*}
    \E_U\big[|W_s|^{\ell}\big]
    &\ \le\ 2^{\ell}\Fe^{\ell/2}\,\E_U\big[Y_s^{\ell/2}\big]
    &&\mbox{raising to the $\ell$-th power}\\
    &\ \le\ 2^{\ell}\Fe^{\ell/2}\,\left( \E_U\big[Y_s^{\ell}\big]\right)^{1/2}
    &&\mbox{Jensen for $x\mapsto x^{2}$}.
  \end{align*}
  Thus, putting this back in \cref{eq:threeterm},
  \begin{align*}
    \E_U\big[(X_{ss})^{\ell}\big]\ & \le\ 3^{\ell-1}\Bigg(\Fe^{\ell}+\ 2^{\ell}\Fe^{\ell/2}\,\left( \E_U\big[Y_s^{\ell}\big]\right)^{1/2}+\E_U\big[Y_s^{\ell}\big]\Bigg) \\
    & \le  3^{\ell-1}\Bigg(\Fe^{\ell}+\ 2^{\ell}\Fe^{\ell/2}\, \left( \frac{4 \ell^2}{d}\right)^{\ell/2} + \left( \frac{4 \ell^2}{d}\right)^\ell \Bigg) & \mbox{by \Cref{cor:fullTD}} \\
    & \leq 3^{\ell-1}\Bigg(\Fe^{\ell}+\ \left( 2 \sqrt{ \Fe {\frac{4 \ell^2}{d}}} \right)^\ell  + \left( \frac{4 \ell^2}{d}\right)^\ell \Bigg)  \\
    & \leq  3^{\ell-1}\Bigg(\Fe+\ \left( 2 \sqrt{ \Fe {\frac{4 \ell^2}{d}}} \right)  + \left( \frac{4 \ell^2}{d}\right) \Bigg)^\ell \\
    & = \ 3^{\ell-1} \left( \sqrt{\Fe} + \sqrt{{\frac{4 \ell^2}{d}}} \right)^{2 \ell} \\
    &  = \  3^{\ell-1}\Bigg(\sqrt{\Fe}+\frac{2\ell}{\sqrt d}\Bigg)^{2\ell}.
  \end{align*}
\end{proof}

\begin{theorem}[Tamper detection]\label{thm:fullTDfinal}
  Let $\alpha < \frac{1}{2}$ and $\delta>0$ be parameters such that the adversarial family of \CPTP maps on $\mathbb{C}^d$, denoted as  $\Fadv$, satisfies the following:
  \begin{itemize}
    \item $ |\Fadv|\le2^{d^{\alpha}}$,
    \item $\Fe(\Phi)\le d^{-\delta}\ \text{ for every }\Phi\in\Fadv $.
  \end{itemize}
  Let  $(\Enc_U,\Dec_U)$ be the Haar random encoding-decoding strategy on $n$ qubits
  (\Cref{def:enc}) with the expansion factor $\gamma$ be as follows:
  \[
    \ \gamma>\max\Big\{\tfrac{2}{1-\alpha},\ \tfrac{1}{\delta},\ \tfrac{1}{1-2\alpha}\Big\}.
  \]
  Then there is $n_0=n_0(\alpha,\delta,\gamma)$ such that for all
  $n\ge n_0$,
  \[
    \Pr_U\big[(\Enc_U,\Dec_U)\text{ is $\varepsilon$-tamper-detecting against }\Fadv\big]\ \ge\ 1-\negl(k),
  \]
  where $\varepsilon = \Theta\left( 2^{-k}\right)$.
\end{theorem}
\begin{proof}
  Throughout, let
  \[
    \varepsilon_0=2^{-k},
    \qquad
    \xi=\min\Big\{\frac{\delta}{2},\ \frac12-\alpha\Big\}>0,
    \qquad
    \zeta=2\xi-\frac1\gamma .
  \]
  By the assumptions on $\gamma$,
  \[
    \frac{1}{\gamma}
    <
    \min\{\delta,1-2\alpha\}
    =
    2\xi,
  \]
  and therefore $\zeta>0$.

  \medskip \noindent
  By \Cref{def:overlap}, $\Pr\big[\Dec_U(\Phi(\Enc_U(s)))\ne\perp\big]=\sum_{t}X_{st}=X_{ss}+\sum_{t\ne s}X_{st}$. We bound each summand by $\varepsilon_0$, except with negligible probability over $U$; this gives $\sum_t X_{st}\le2\varepsilon_0$ for all $s$ and $\Phi$ simultaneously, i.e.\ the code is $2\varepsilon_0$-tamper-detecting with $2\varepsilon_0=\Theta(2^{-k})$.

  \medskip\noindent
  \emph{Off-diagonal branch.} Since $\alpha<\tfrac12<1$ and $\gamma>\tfrac{2}{1-\alpha}$, the hypotheses of \Cref{thm:utd} hold. Hence there is $n_0^{\mathrm{off}}=n_0(\alpha,\gamma)$ such that for all $n\ge n_0^{\mathrm{off}}$, except with probability $\negl(k)$ over $U$, the code is $\varepsilon_0$-relaxed-tamper-detecting against $\Fadv$; by \Cref{lem:dropdiag} this says exactly that
  \[
    \sum_{t\ne s}X_{st}\ \le\ \varepsilon_0
    \qquad\text{for every }s\in\mathcal M\text{ and every }\Phi\in\Fadv .
  \]

  \medskip\noindent
  \emph{Diagonal branch.} Fix the moment order $\ell=\lceil d^{\alpha}\rceil$, so that $2\le\ell\le2d^{\alpha}$ since $\alpha>0$. Applying Markov's inequality at order $\ell$ together with \Cref{lem:diagmoment},
  \[
    \Pr_U\big[X_{ss}>\varepsilon_0\big]
    \ \le\ \frac{\E_U\big[(X_{ss})^{\ell}\big]}{\varepsilon_0^{\ell}}
    \ \le\ R^{\ell},
    \qquad\text{where}\qquad
    R=\frac{3\big(\sqrt{\Fe}+2\ell/\sqrt d\big)^{2}}{\varepsilon_0}.
  \]
  We first bound $R$. Since $\Fe\le d^{-\delta}$ we have $\sqrt{\Fe}\le2^{-\delta n/2}\le2^{-\xi n}$, and since $\ell\le2d^{\alpha}$ we have $2\ell/\sqrt d\le4\cdot2^{(\alpha-\frac12)n}\le4\cdot2^{-\xi n}$; hence $\sqrt{\Fe}+2\ell/\sqrt d\le5\cdot2^{-\xi n}$. With $\varepsilon_0=2^{-n/\gamma}$,
  \[
    R\ \le\ 3\cdot25\cdot2^{-2\xi n}\cdot2^{\,n/\gamma}\ =\ 75\cdot2^{-\zeta n},
    \qquad\text{so}\qquad
    \log_2 R\ \le\ 7-\zeta n .
  \]
  Taking a union bound over the $2^{k}$ messages and the at most $2^{d^{\alpha}}$ channels, and using $\ell\ge d^{\alpha}$ together with $7-\zeta n\le0$ for $n\ge7/\zeta$,
  \begin{align*}
    \log_2\Pr_U\Big[\exists\,s,\Phi\ :\ X_{ss}>\varepsilon_0\Big]
    &\ \le\ k+d^{\alpha}+\ell\,\log_2 R
    &&\mbox{the union bound}\\
    &\ \le\ \frac{n}{\gamma}+d^{\alpha}+d^{\alpha}\big(7-\zeta n\big)
    &&\mbox{$k=\frac{n}{\gamma}$, $\ell\ge d^{\alpha}$}\\
    &\ =\ -\,\zeta\,n\,d^{\alpha}+8d^{\alpha}+\frac{n}{\gamma}.
  \end{align*}
  It remains to absorb the two positive terms.
  Define
  \[
    n_0^{\mathrm{diag}}=\max\Big\{\frac{32}{\zeta},\ \frac1\alpha\log_2\frac{4}{\gamma\,\zeta}\Big\},
  \]
  and let $n\ge n_0^{\mathrm{diag}}$.
  From $n\ge32/\zeta$ we get $8d^{\alpha}\le\tfrac14\zeta n\,d^{\alpha}$, and from $n\ge\tfrac1\alpha\log_2\tfrac4{\gamma\zeta}$ we get $d^{\alpha}=2^{\alpha n}\ge4/(\gamma\zeta)$, which on multiplying both sides by $\tfrac14\zeta n$ gives $\tfrac n\gamma\le\tfrac14\zeta n\,d^{\alpha}$.
  Substituting these two bounds,
  \begin{align*}
    \log_2\Pr_U\Big[\exists\,s,\Phi\ :\ X_{ss}>\varepsilon_0\Big]
    &\ \le\ -\zeta n\,d^{\alpha}+\tfrac14\zeta n\,d^{\alpha}+\tfrac14\zeta n\,d^{\alpha}\\
    &\ =\ -\tfrac12\,\zeta\,n\,d^{\alpha}.
  \end{align*}

  \medskip\noindent
  Setting $n_0=\max\{n_0^{\mathrm{off}},\,n_0^{\mathrm{diag}}\}$, both branches fail with probability at most $\negl(k)$ for $n\ge n_0$. Hence, except with probability $\negl(k)$ over $U$, we have $\sum_t X_{st}\le2\varepsilon_0$ for every $s\in\mathcal M$ and every $\Phi\in\Fadv$, so $(\Enc_U,\Dec_U)$ is $\varepsilon$-tamper-detecting with $\varepsilon=2\varepsilon_0=\Theta(2^{-k})$.
  Since $n_0$ depends only on $\alpha$, $\delta$ and $\gamma$, the expansion factor satisfies $\gamma=\mathcal{O}(1)$ whenever $\alpha$ and $\delta$ are constants.
\end{proof}
    \section{Revocable keyless encryption}\label{sec:rev}

A revocable encryption scheme lets the sender demand the ciphertext back: if the returned state passes verification, the recipient is left with no information about the message, even if he becomes computationally unbounded afterwards~\cite{Unr15}.
We show that the scheme of \Cref{thm:pq-main}, without any modification, is such a scheme, and that it achieves guarantees analogous to those of the revocable timed-release encryption of~\cite{Unr15}.

\subsection{Depth-bounded adversaries}\label{sec:depth}

We measure elapsed time in circuit layers: a recipient executing $R$ layers per unit time has applied a circuit of depth at most $p$ by time $p/R$.
Depth is the right measure here, because a timing assumption must be invariant under parallelization: a gate count is not, since a recipient with many processors executes many gates per unit time.
However, our results are valid also for gate count as a measure of time.

\begin{definition}[Depth-bounded phase-1 circuits]\label{def:depth}
  For $p,w\in\N$, let
  \[
    \cD_{p,w}
    =
    \Big\{\,V:\C^d_C\otimes\ket{0}_{E_0}\to\C^d_C\otimes\Hs_E
    \;:\;V\text{ is implemented by a $\G$-circuit of depth at most }p\,\Big\},
  \]
  where the circuit acts on the $n$ qubits of $C$ together with at most $w$ ancillas.
\end{definition}

Throughout this section, $\cD_{p,w}$ replaces $\cG_p$ in \Cref{def:adversary} and in \cref{def:momentorder}.
This is the only point at which the adversary class enters the analysis: every later step of the proof of \Cref{thm:pq-main}.

\begin{lemma}[Depth counting]\label{lem:count}
  The number of phase-1 circuits of depth at most $p$ on $n+w$ qubits satisfies
  \[
    \log_2|\cD_{p,w}|=O\big(p\,(n+w)\log(n+w)\big).
  \]
  Hence, \cref{def:momentorder} satisfies with $\ell=O\big(n+p(n+w)\log(n+w)\big)$, which is $\poly(n,p,w)$.
\end{lemma}
\begin{proof}
  Put $N=n+w$.
  A layer is specified by a partition of the $N$ wires into blocks of size at most $c$, together with a choice of gate from $\G$ for each block.
  Each wire selects one of at most $N$ blocks, and there are at most $N$ blocks, so the number of distinct layers is at most $N^{N}|\G|^{N}=2^{O(N\log N)}$.
  A circuit of depth at most $p$ is a sequence of at most $p$ layers, and different circuit descriptions may implement the same isometry, so
  \[
    \log_2|\cD_{p,w}|\le p\,N\log_2\big(|\G|N\big)=O\big(p\,(n+w)\log(n+w)\big).
  \]
  The bound on $\ell$ then follows from \cref{def:momentorder}.
\end{proof}

A phase-1 circuit with $w$ ancillas induces a channel $\Phi_V$ of Kraus rank up to $2^{w}$, and $w$ is a free parameter here, not tied to the depth $p$.
The terms in the proof of \Cref{thm:pq-main} are independent of that rank (unlike proof of \cite{BKR26}), the contributions of the slices being combined through \cref{eq:slice-tp} rather than counted term by term, as observed in \Cref{subsec:moment}.
A bound carrying a factor $r^{\ell}$ in the Kraus rank, of the kind discussed after \Cref{lem:rankfree}, would be vacuous for all but logarithmic $w$, so wide holders are admissible only because no such factor appears.

\subsection{The revocation experiment}\label{sec:expt}

A \emph{holder} is a two-phase adversary $\A=(\A^{\mathrm{msg}},V,\Lambda)$ in the sense of \Cref{def:adversary}, with $V=V(\pp)\in\cD_{p,w}$.
The isometry $V$ is everything the holder does while he holds the codeword, and the unbounded channel $\Lambda$ is what he does after returning it.
A \emph{verification procedure} $\Ver(\pp,\cdot)$ is a two-outcome measurement on $C$, with outcomes $\top$ (accept) and $\perp$ (reject).

Fix a scheme $\Pi$ with verifier $\Ver$, a holder $\A$, and $b\in\{0,1\}$.
Let $G$ be a classical register recording the verifier's outcome, and let $P$ be a classical register holding the public string $\pp$.

\begin{center}
  \fbox{
    \begin{minipage}{0.92\textwidth}
      $\Expt^{\REV}_{\Pi,\A}(\lambda,k,b)$:
      \begin{enumerate}\itemsep2pt
        \item $\pp\leftarrow\Setup(1^\lambda,1^k)$.
        \item $(m_0,m_1,\rho_S)\leftarrow\A^{\mathrm{msg}}(\pp)$.
        \item $\rho_C\leftarrow\Enc(\pp,m_b)$.
        \item $\rho_{CE}\leftarrow V\rho_C\adj V$.
        \item The holder returns $C$, and $g\leftarrow\Ver(\pp,\rho_C)$.
        \item Output the state $\big(\pp,\ g,\ \Lambda(\rho_E),\ \rho_S\big)$ on $P\otimes G\otimes E'\otimes S$.
      \end{enumerate}
    \end{minipage}}
\end{center}
Before the phase-2 channel $\Lambda$ is applied, the joint state on $G\otimes E$ is a cq state of the form
\begin{equation}\label{eq:revstate}
  \sigma^{\REV}_{\pp,V}(m)\;=\;\sum_{g\in\{\top,\perp\}}\proj{g}_G\otimes\sigma^{\REV}_g(m),
\end{equation}
where each $\sigma^{\REV}_g(m)$ is a subnormalized state on $E$ and $\Tr\sigma^{\REV}_\top(m)$ is the probability that the returned register is accepted when $m$ is encoded.
We write
\begin{equation}\label{eq:revacc}
  \tilde\sigma^{\REV}_{\pp,V}(m)\;=\;\proj{\top}_G\otimes\sigma^{\REV}_\top(m)
\end{equation}
for the restriction of \cref{eq:revstate} when the verifier accepts.

\begin{definition}[Revocable keyless encryption]\label{def:rev}
  A coding scheme $\Pi$ with a verification procedure $\Ver$ is a \emph{$(p,q)$-revocable keyless encryption scheme against width $w$} if the following three conditions hold.

  \medskip\noindent\textbf{(R) Release.} $\Dec(\pp,\cdot)$ is implemented by a circuit of depth at most $q(\lambda,k)$, and
  \[
    \Pr\big[\Dec(\pp,\Enc(\pp,m))=m\big]=1
  \]
  for every $\lambda,k$, every $\pp$ in the image of $\Setup(1^\lambda,1^k)$ and every $m\in\M$.

  \medskip\noindent\textbf{(RC) Revocation completeness.} An honestly returned codeword is accepted:
  \[
    \Pr\big[\Ver(\pp,\Enc(\pp,m))=\top\big]=1
  \]
  for every $\lambda,k$, every $\pp$ in the image of $\Setup(1^\lambda,1^k)$ and every $m\in\M$.

  \medskip\noindent\textbf{(RH) Revocable hiding.} For every holder $\A$ with $V\in\cD_{p,w}$,
  \[
    \E_{\pp}\Big\|\tilde\sigma^{\REV}_{\pp,V}(m_0)-\tilde\sigma^{\REV}_{\pp,V}(m_1)\Big\|_1\;\le\;\negl(\lambda),
  \]
  where $V=V(\pp)$ and $(m_0,m_1)$ are those chosen by $\A$ on input $\pp$.
\end{definition}

Only the accepting branch is constrained.
On the rejecting branch the holder may keep the codeword, so no guarantee is possible there; what revocation provides is that the sender learns the attempt failed.
The definition of~\cite{Unr15} is restricted in the same way.

The verifier outputs a single bit and no decoded message, so no analogue of the coarsening $\same$ of \Cref{sec:our_primitive} is needed.
Since $g$ is released to the distinguisher, the acceptance probability must itself be message-independent up to $\negl(\lambda)$, and the trace distance in (RH) imposes this implicitly, exactly as clause \textup{(ES)} of \Cref{def:pq} does for the non-abort probability.

\subsection{The scheme of \texorpdfstring{\Cref{thm:pq-main}}{Theorem 1} is revocable}\label{sec:revthm}

\begin{fact}[\cite{Wat18}]\label{fact:dpi}
  $\|\Phi(X)\|_1\le\|X\|_1$ for every quantum channel $\Phi$ and every Hermitian operator $X$.
\end{fact}

\begin{lemma}\label{lem:coarse}
  Fix $\pp$ and a phase-1 isometry $V$, and let $\Ver(\pp,\cdot)$ be the measurement $\Dec(\pp,\cdot)$ followed by the map sending every $\hat m\in\M$ to $\top$ and $\perp$ to $\perp$.
  Let $\chi:\{\same\}\cup\M\cup\{\perp\}\to\{\top,\perp\}$ be given by $\chi(\perp)=\perp$ and $\chi(f)=\top$ otherwise, acting on the classical register $F$.
  Then, for every $m\in\M$,
  \[
    \sigma^{\REV}_{\pp,V}(m)=(\chi\otimes\id_E)\big(\sigma_{\pp,V}(m)\big),
    \qquad
    \tilde\sigma^{\REV}_{\pp,V}(m)=(\chi\otimes\id_E)\big(\tilde\sigma_{\pp,V}(m)\big),
  \]
  with $\sigma_{\pp,V}$ and $\tilde\sigma_{\pp,V}$ as in \cref{eq:branch-decomp,eq:nonabort}.
  Thus, for all $m_0,m_1\in\M$,
  \[
    \Big\|\tilde\sigma^{\REV}_{\pp,V}(m_0)-\tilde\sigma^{\REV}_{\pp,V}(m_1)\Big\|_1
    \;\le\;
    \Big\|\tilde\sigma_{\pp,V}(m_0)-\tilde\sigma_{\pp,V}(m_1)\Big\|_1 .
  \]
\end{lemma}
\begin{proof}
  Steps 1--4 of $\Expt^{\REV}_{\Pi,\A}$ and $\Expt^{\NTDshort}_{\Pi,\A}$ coincide, and in both experiments the register $C$ is then measured with the decoder's measurement $\{\Pi_s\}_{s\in\M}\cup\{\Pi_\perp\}$.
  They differ only in how the outcome is recorded: $\Expt^{\NTDshort}_{\Pi,\A}$ records $f$, and $\Expt^{\REV}_{\Pi,\A}$ records $g=\top$ if $\hat m\in\M$ and $g=\perp$ otherwise.
  Since $f=\perp$ exactly when $\hat m=\perp$, we have $g=\chi(f)$, and therefore
  \[
    \sigma^{\REV}_{\pp,V}(m)=\sum_f\proj{\chi(f)}_G\otimes\sigma_f(m)=(\chi\otimes\id_E)\big(\sigma_{\pp,V}(m)\big),
  \]
  which gives the first identity.
  The branch $g=\top$ is the union of the branches $f\ne\perp$, which gives the second identity.
  The map $\chi\otimes\id_E$ is a quantum channel, classical on $F$, and is the same channel for both messages, so the inequality follows from \Cref{fact:dpi} applied to the Hermitian operator $\tilde\sigma_{\pp,V}(m_0)-\tilde\sigma_{\pp,V}(m_1)$.
\end{proof}

\begin{theorem}[Revocable keyless encryption]\label{thm:rev}
  Let $\Pi$ be the scheme constructed in \Cref{thm:pq-main} (with $\cD_{p,w}$ in place of $\cG_p$), and let $\Ver:=\Dec_U$, accepting if and only if the outcome lies in $\M$.
  Then $(\Pi,\Ver)$ is a $(p,q)$-revocable keyless encryption scheme against width $w$, where $q$ is the depth of $\Dec_U$.
  The error in (RH) is $2^{-\Omega(n)}$.
\end{theorem}
\begin{proof}
  (R) Perfect completeness holds for every $U$ in the support of the design, by the completeness of \Cref{thm:pq-main}.
  The depth of $\Dec_U$ is at most its size, which is $\poly(n,p)$ by \cref{eq:designsize}, so $q=\poly(n,p)$; since the designs of~\cite{MPSY24} are of linear depth, a sharper bound can be read off from that construction.

  \medskip\noindent
  (RC) By (R), $\Dec_U(\Enc_U(m))=m\in\M$ with probability one, so $\Ver$ accepts with probability one.

  \medskip\noindent
  (RH) Let $\A$ be a holder with phase-1 isometry $V=V(\pp)\in\cD_{p,w}$.
  For every $\pp$, \Cref{lem:coarse} bounds the quantity in (RH) by the quantity in clause \textup{(ES)} of \Cref{def:pq}, for the same $V$ and the same pair of messages.
  By \Cref{lem:count}, the moment order \cref{def:momentorder} remains $\poly(n,p,w)$ when $\cG_p$ is replaced by $\cD_{p,w}$, so the bound established in \Cref{sec:secrecy} is available for the class $\cD_{p,w}$.
  Taking the expectation over $\pp$ gives
  \[
    \E_{\pp}\Big\|\tilde\sigma^{\REV}_{\pp,V}(m_0)-\tilde\sigma^{\REV}_{\pp,V}(m_1)\Big\|_1
    \;\le\;
    \E_{\pp}\Big\|\tilde\sigma_{\pp,V}(m_0)-\tilde\sigma_{\pp,V}(m_1)\Big\|_1
    \;\le\;2^{-\Omega(n)} .
  \]
  The good events of \Cref{sec:secrecy} hold uniformly over the whole class, so the bound remains valid when $V$ is chosen as a function of $\pp$.
\end{proof}

A verifier who remembers $m_b$ may instead measure $\{\Pi_{m_b},\id-\Pi_{m_b}\}$.
The accepting branch then shrinks to the branch $f=\same$, and the quantity in (RH) becomes the term $(\mathrm I)$ of \cref{eq:es-split}, which is bounded in \Cref{sec:secrecy}.

    \paragraph{AI Disclosure.} No AI tools were used in the preparation of this manuscript. The only exception is the use of AI-assisted copy-editing for minor grammar and clarity improvements to text originally written by the authors.

    \paragraph{Acknowledgments.} 
    A.\,B.~and U.\,K~acknowledge the support of the Natural Sciences and Engineering Research Council of Canada (NSERC) (ALLRP-578455-2022, RGPIN-2022-05167), and of the Canada Research Chairs Program (CRC-2023-00173). D.\,R.~acknowledges funding by the ANR for the JCJC grants LINKS (No. ANR-23-CE47-0003), the T-ERC QNET (No. ANR24-ERCS-0008), the project QUANTINT, as well as the European Union's Horizon 2020 Research and Innovation Programme under QuantERA Grant Agreements No.~731473 and No.~101017733. 
    
    \bibliographystyle{bibtex/bst/alphaarxiv.bst}
    \bibliography{bibtex/bib/quasar-full.bib,
        bibtex/bib/quasar.bib,
        bibtex/bib/quasar-more-tamper-new.bib,
        bibtex/bib/quasar-more-merged.bib,
        bibtex/bib/quasar-more.bib}
\fi

\end{document}